\documentclass[11pt,superscriptaddress,aps,pra]{revtex4}
\usepackage{amssymb,amsfonts,amsmath,amsthm}

\let\Cite\cite
\def\Fig{FIG.}
\def\PICWIDTH{0.65\textwidth}

\usepackage[utf8]{inputenc}
\usepackage{bbm}

\usepackage{graphicx}
\usepackage{color}
\usepackage{transparent}
\usepackage{picins}

\def\currenttime{\hour=\time \divide\hour by 60 \number\hour:%
  \multiply\hour by 60 \minute=\time \global\advance\minute by -\hour%
  \ifnum\minute<10 0\number\minute\else\number\minute\fi}

\newcommand{\Now}{\number\day\ {\ifcase\month\or%
 January\or February\or March\or April\or May\or June\or%
 July\or August\or September\or October\or November\or December\fi}\
\number\year\ at \currenttime}

\newcommand{\E}{\mathrm{e}}
\newcommand{\I}{\mathrm{i}} 

\newcommand{\Real}{\mathop{\mathrm{Re}}}
\newcommand{\Imag}{\mathop{\mathrm{Im}}}

\newcommand{\ket}[1]{{\left|#1\right\rangle}}

\newcommand{\braket}[2]{{\left\langle#1|#2\right\rangle}}
\newcommand{\ketbra}[2]{{\left|#1\right\left\bigl\langle #2\right|}}

\newcommand{\WITH}{\quad\mbox{with}\quad}
\newcommand{\AND}{\quad\mbox{and}\quad}

\newtheorem{thm}{Theorem}[section]

\newtheorem{lem}[thm]{Lemma}
\newtheorem{prop}[thm]{Proposition}
\newtheorem{cor}[thm]{Corollary}

\newtheorem{conjecture}[thm]{Conjecture}

\def\idty{{\leavevmode\rm 1\mkern -5.4mu I}} 
\def\idty{{\mathbbm1}} 

\def\Rl{{\mathbb R}}\def\Cx{{\mathbb C}}

\def\norm #1{\Vert #1\Vert}

\def\brAAket#1#2#3{\langle#1\vert#2\vert#3\rangle}

\def\ket #1{\vert#1\rangle}
\def\ketbra #1#2{{\vert#1\rangle\langle#2\vert}}
\def\kettbra#1{\ketbra{#1}{#1}}

\def\tr{\mathop{\rm tr}\nolimits}
\def\abs#1{\vert#1\vert}

\def\farray#1#2#3#4{\left#1\begin{array}{#2}#3\end{array}\right#4}
\def\Cases#1{\farray\lbrace{cl}{#1}.}

\let\veps\varepsilon

\def\inv{^{-1}}
\let\eless\sqsubseteq  
\let\antipode\blacktriangledown

\newcommand\ph{\varphi}

\def\phhi{\widehat\ph}
\def\GG{{\mathcal G}}
\def\iir{{\mathrm i}r}

\def\Renyi{R{\'e}nyi}
\def\sigQual{\sigma_{\rm eq}}
\def\bic{\zeta}  

\begin{document}

\title{Optimality of entropic uncertainty relations}

\author{Kais Abdelkhalek}
\affiliation{Institut f\"{u}r Theoretische Physik, Leibniz Universit\"{a}t Hannover, Germany}
\author{Ren\'{e} Schwonnek}
\affiliation{Institut f\"{u}r Theoretische Physik, Leibniz Universit\"{a}t Hannover, Germany}
\author{Hans Maassen}
\affiliation{Department of Mathematics, Radboud University, Nijmegen, The Netherlands}
\author{Fabian Furrer}
\affiliation{Department of Physics,  University of Tokyo, Japan}
\author{J\"{o}rg Duhme}
\affiliation{Institut f\"{u}r Theoretische Physik, Leibniz Universit\"{a}t Hannover, Germany}
\author{Philippe Raynal}
\affiliation{Centre for Quantum Technologies, National University of Singapore, Singapore},
\affiliation{University Scholars Programme, National University of Singapore, Singapore}
\author{Berthold-Georg Englert}
\affiliation{Centre for Quantum Technologies, National University of Singapore, Singapore}
\affiliation{Department of Physics, National University of Singapore, Singapore}
\affiliation{MajuLab, CNRS-UNS-NUS-NTU International Joint Unit, UMI 3654, Singapore}
\author{Reinhard F. Werner}
\affiliation{Institut f\"{u}r Theoretische Physik, Leibniz Universit\"{a}t Hannover, Germany}

\begin{abstract}\vskip40pt\noindent
The entropic uncertainty relation proven by Maassen and Uffink for arbitrary pairs of two observables is known to be non-optimal. Here, we call an uncertainty relation optimal, if the lower bound can be attained for any value of either of the corresponding uncertainties. In this work we establish optimal uncertainty relations by characterising the optimal lower bound in scenarios similar to the Maassen-Uffink type. We disprove a conjecture by Englert {\it et al.} and generalise various previous results. However, we are still far from a complete understanding and, based on numerical investigation and analytical results in small dimension, we present a number of conjectures.
\end{abstract}


\maketitle

\section{Introduction}%

As a characteristic trait, quantum systems possess  properties that are incompatible --- properties that are equally real but mutually exclusive. In a pair of incompatible properties, if we have precise knowledge about one
property, what we know about the other is necessarily imprecise. More generally, we can trade knowledge about one property for knowledge about the other and so know both imperfectly, and quantify the
lack of knowledge by a suitable measure of uncertainty.
Then, the compromises allowed by nature have their mathematical expressions in
the form of \emph{uncertainty relations}, which are inequalities that follow
from the formalism of quantum theory.

The study of uncertainty tradeoffs originated in Heisenberg's pioneering work\cite{Heisenberg1927} of
1927 and was soon brought into a clear mathematical form by Kennard\cite{Kennard1927}.
Weyl gave another early proof\cite{Weyl1928}. He was apparently unaware of Heisenberg's paper and gives credit for the idea to Pauli, who seems to have learned it from Heisenberg in a letter prior to the publication of \Cite{Heisenberg1927}. The modern textbook proof combining the Schwarz inequality with the commutation relations is due to Robertson\cite{Robertson1929}.
In this tradition the ``uncertainty of observable $X$ in the state $\rho$'' is identified with the root of the variance of the probability distribution of the outcomes of an $X$-measurement on particles prepared according to $\rho$, i.e.,
\begin{equation}\label{eq:1.1}
  \delta X=\sqrt{\tr{(\rho X^2)}-\tr(\rho X)^2}\,,
\end{equation}
The key requirement for Heisenberg's  uncertainty relation $\delta Q\,\delta P\geq\hbar/2$ to hold is that these variances are evaluated in the same state. The relation is thus a quantitative expression of the observation that there are no dispersion-free states, and is hence of the type ``preparation uncertainty relation''. This is in contrast to ``measurement uncertainty relations'' which express the feature of quantum mechanics that some observables may not be measured jointly, which also implies that any measurement of one observable $X$ implies a disturbance of the other in the sense that it cannot be inferred from a measurement on the state after an $X$-measurement. This aspect, although more prominent in Heisenberg's paper than the preparation side, was made precise only recently\cite{Busch2014} (also \Cite{ColesFurrer2013,RenesScholz2014}).

In this paper we are interested in preparation uncertainty relations for quantum systems of finite dimension $d$. A standard scenario in which this is of interest is the tradeoff between Welcher-Weg information and interference patterns at a multi\-port interferometer.
In this minimalistic instance of  wave-particle duality\cite{englert2008} one observable would detect particles on each of the internal paths of the interferometer, thus detecting a particle-like property, whereas the detectors at the end pick up wave-like interference. Uncertainty in this situation expresses the physical fact that if we prepare incoming particles so that they all go along the same path, we loose the interference contrast and, conversely, that large interference contrast is only possible when all paths are ``traversed'' with roughly equal probability. Another standard context for finite-dimensional uncertainty relations is quantum information theory, particularly quantum key distribution. Large parts of this theory have been developed in finite dimension, and there are many situations in which the incompatibility as expressed by uncertainty relations plays an important role (e.g. in security proofs\cite{Tomamichel2012} of cryptographic protocols).

What is common to these motivating instances of finite-dimensional uncertainty is that the outcomes of the respective observables are labelled in a completely arbitrary way. However, a variance depends not only on the abstract outcomes and their probabilities, but also on the real numbers we assign to them. For example, by multiplying all these numbers by the same factor we also multiply $\delta X$. Moreover, variance will change if we permute the outcomes, which is as easy to do with beams in optical fibers as with freely re-codable bits of information. Basically motivated by such considerations, Deutsch\cite{Deutsch1983} suggested to use entropies to quantify the (lack of) sharpness of a probability distribution. This led to the famous entropic uncertainty relation established by Maassen and Uffink\cite{Maassen1988}, to which we will refer to as the {\it MU bound}.
It describes the sharpness tradeoff for the outcome distributions $p_X^\rho$ and $p_Y^\rho$ of two observables
$X,Y$ in the same state $\rho$ in terms of their \Renyi\ entropies $H_\alpha$, $H_\beta$ (see \eqref{eq:defentropy}), provided that these parameters satisfy the {\it duality relation}
\begin{equation}\label{eq:duality}
\frac{1}{\alpha}+\frac{1}{\beta}=2 \ .
\end{equation}
When the observables $X$ and $Y$ are given in terms of their eigenbases $\{x_i\}$ and $\{y_j\}$, so that
$p_X^\rho(i)=\brAAket{x_i}\rho{x_i}$ and $p_Y^\rho(j)=\brAAket{y_j}\rho{y_j}$, the MU bound is
\begin{equation}\label{eq:MUbound}
  H_{\alpha}(p_X^\rho)+H_{\beta}(p_Y^\rho)\geq
   -\log \max_{j,k}\abs{\braket{y_k}{x_j}}^2\,.
\end{equation}
The bound becomes zero when the two bases share a vector, and maximal (namely $\log d$) if the bases are mutually unbiased, so that all scalar products $\braket{y_k}{x_j}$ have the same modulus.

An alternative to entropies would again be variances, once one realizes that for defining a variance it is not really necessary to have $\Rl$-valued random variables. It suffices to have outcomes in a metric space $\Omega$ with metric $\Delta$, so that the variance of a probability measure $\mu$ on $\Omega$ becomes
\begin{equation}\label{varianceMetric}
  \mbox{var}(\mu)=\inf_{\eta\in\Omega}\int\mu(d\omega)\ \Delta(\omega,\eta)^2 \,.
\end{equation}
When $\Omega=\{1,\ldots,d\}$ the only permutation invariant metrics are $\Delta(i,j)=c(1-\delta_{ ij})$, and we will just fix the constant $c=1$.
Then
\begin{equation}\label{varianceDiscrete}
  \mbox{var}(p)=\min_{j}\sum_ip(i)\,(1-\delta_{ ij})^2=1-\max_jp(j) \,.
\end{equation}
Up to a rescaling this is the so-called min entropy $H_\infty(p)=-\log \max_j p(j)$.

How then should we write an uncertainty relation in this general context? We will see that it is not wise to fix in advance the functional form of the tradeoff relation between $H_{\alpha}(p_X)$ and $H_{\beta}(p_Y)$. Instead, the best and most intuitive representation of the tradeoff is the diagram of all pairs $(H_{\alpha}(p_X),H_{\beta}(p_Y))$, ranging over all choices of input states $\rho$.
\begin{figure}[h]
	\def\svgwidth{0.95\textwidth}
     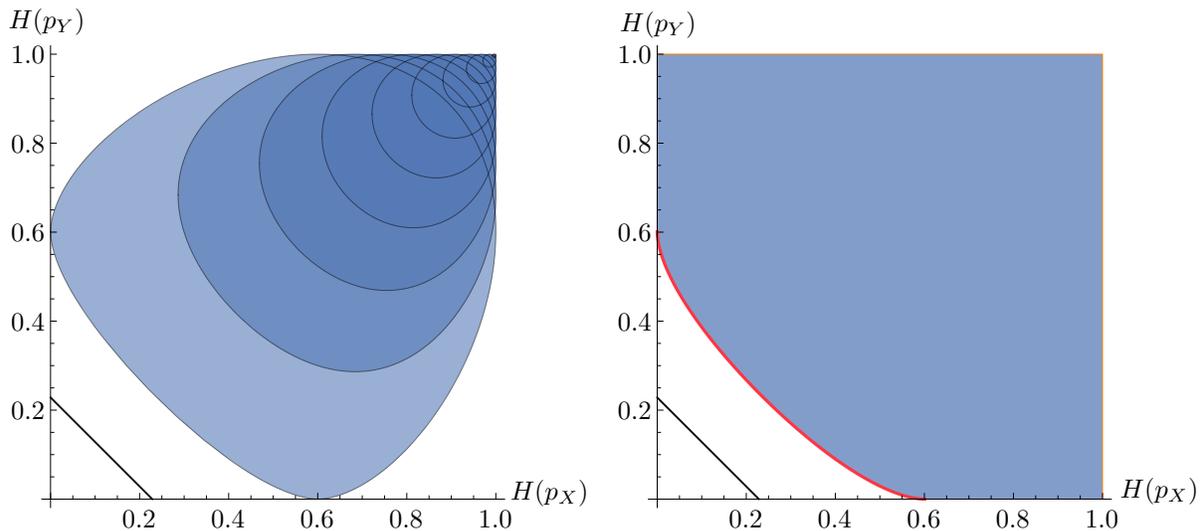
	\caption{Entropy pairs for $d=2$ and the observables $X=\sigma_z$ and $Y=(\sigma_x+\sigma_z)/\sqrt 2$. Left panel: The shaded set gives all pairs $(H(p_X^\rho),H(p_Y^\rho))$. The contours describe the subsets which can be reached by pure states with a fixed admixture of $\rho=\idty/2$. Right panel:
		The shaded set is the ``monotone closure'' of the one on the left (see text). The solid curve represents the optimal bound: For entropy pairs on this bound it is impossible to reduce one entropy without enlarging the other. The thin line closer to the origin is the MU bound. }
	\label{fig:qbits}
\end{figure}
An advantage of this representation is also that it changes in a simple way by a rescaling like the replacement of the variance (\ref{varianceDiscrete}) by $H_\infty$. For qubits ($d=2$), all measures of sharpness are functions of each other, so all such diagrams are equivalent. Figure~\ref{fig:qbits} is drawn for the Shannon entropy $H=H_1$. Some details of the diagram of {\it all} pairs of entropies, shown on the left, are clearly not relevant for the uncertainty tradeoff, in which we ask {\it how small} we can simultaneously make the entropies. For this question it is the lower left corner of the diagram which matters, i.e., the set in the right diagram.
It can be described as adding to any pair of entropies the full closed positive (north-east) quadrant above it. It is completely described by its lower left boundary, consisting of those entropy pairs with the property that for no other state one can have one entropy strictly smaller and the other at least as small. We consider the resulting curve as the complete description of the uncertainty tradeoffs between the entropies involved. Characterising this curve is the aim of this paper.

We will always consider a quantum system in a $d$-dimensional Hilbert space, and consider two projection valued observables with $d$ outcomes. This amounts to the choice of two bases $\{x_i\}$ and $\{y_j\}$, and for the question at hand the choice is completely described by the unitary overlap matrix $U_{ij}=\braket{x_i}{y_j}$ modulo multiplication by diagonal unitary matrices or permutation matrices from either side. In the motivating standard case, closest to the case of position and momentum of continuous variables, the $U$ represents the discrete Fourier transform of either the cyclic group of $n$ elements or, if $n$ is composite, another finite abelian group of order $n$. More generally, we also consider complex Hadamard matrices, i.e., unitary operators such that $\abs{U_{ ij}}=1/\sqrt d$ for all $i,j$. The bases are then called mutually unbiased, and we can think of a multiport interferometer generalizing a $50{:}50$-beam splitter. Such bases also represent complementary pairs of measurements from the informational point of view. However, we will not restrict our study to these special classes of unitary matrices --- several results will hold for arbitrary unitary matrices. For generalized observables (POVMs) or $k$-tuples of observables similar questions can be asked, but we will not consider them in this paper. For the quantification of uncertainty or unsharpness we use the \Renyi\ entropies $H_\alpha$ ($1/2\leq\alpha\leq\infty$), and denote by $H=H_1$  the standard case of the Shannon entropy. Mostly we assume that the \Renyi\ parameters $\alpha$ and $\beta$ used for $X$ and $Y$, respectively, satisfy the duality relation \eqref{eq:duality}. Again, the questions make sense also for other measures, e.g., related to majorization, or for variances, but these will not be considered here. We will also restrict ourselves to state-independent bounds, i.e., to the entropy pairs achievable by arbitrary states. When more is known about the state, for example about further expectation values, the entropy diagram for the subset may be quite different. Thus we do not consider inequalities like the Robertson inequality for variances, where the lower bound depends on the expectation of a commutator.

{\it Outline.} In Sect.~\ref{sec:prelims} we briefly define all the relevant quantities and state our problem in precise mathematical terms. We present a brief review of previous results in Sect.~\ref{sec:previ}. In Sect.~\ref{sec:equality} we provide a characterization of the case of equality in the MU bound and thereby show that the MU bound is not optimal in almost all cases. Our main results are presented in Sect.~\ref{sec:main}. We are not able to completely solve the problem in all its generality. However, we provide strong conjectures (Sect.~\ref{sec:conj}) which, if true, heavily reduce the complexity of the problem.

\section{Preliminaries and Notation}\label{sec:prelims}
For $\alpha\in[\frac12,\infty]$ the {\it $\alpha$-\Renyi\ entropy} of a probability distribution $p\in(0,1)^d$ is defined by
\begin{equation}\label{eq:defentropy}
H_{\alpha}(p)=\Cases{\frac{1}{1-\alpha}\log \sum_{j=1}^{d} p(j)^{\alpha} &\mbox{ if\ }\alpha\neq1,\infty\\
                     -\sum_{j=1}^{d} p(j)\log p(j) &\mbox{ if\ }\alpha=1 \mbox{\Huge\strut}\\
                     -\log \max_j p(j)  &\mbox{ if\ }\alpha=\infty\,.    \mbox{\Huge\strut}}
\end{equation}
The logarithms can be taken in any base (as long as it is always the same base). We follow the information theory convention of using base\,-$2$ logarithms, although base $d$ would also be natural in our context, as it would normalize the range to $0\leq H_\alpha(p)\leq \log d=1$. Monotone functions of the entropies tell the same story. In this sense we also cover ``Tsallis entropies'' $T_\alpha(p)=(1-\alpha)\inv(1-\sum_jp(j)^\alpha)$.

Each entropy diagram will be drawn for a fixed choice of observables (i.e., bases) $X,Y$ and values of the \Renyi\ parameters $\alpha,\beta$, so that we consider a map $f$ from the state space to $\Rl_+^2$ given by
\begin{equation}\label{deff}
  f(\rho)=\bigl(f_1(\rho),f_2(\rho)\bigr) =\bigl(H_\alpha(p_X^\rho),H_\beta(p_Y^\rho)\bigr)\,.
\end{equation}
For any choice we can define the order relation $\eless$ on the state space, so that $\rho\eless\rho'$ stands for ``$f_1(\rho)\leq f_1(\rho')$ and $f_2(\rho)\leq f_2(\rho')$''. The {\it uncertainty diagram} is the monotone closure of the range $\{f(\rho)\}$, i.e., it is the set $S$ containing precisely the pairs $(h_1,h_2)\in S$ for which there is a state $\rho$ with $f_i(\rho)\leq h_i$ for $i=1,2$ (compare \Fig~\ref{fig:qbits}). We call a state $\rho$ {\it optimal} if $\rho'\eless\rho$ implies $\rho\eless\rho'$, and hence $f(\rho)=f(\rho')$. The corresponding {\it optimal points} in the entropy plane are characterized by the property that the uncertainty diagram contains no points to their south-west. We call the set of all optimal points the curve of minimal entropies or the {\it optimal bound}. Therefore the optimal bound corresponds to a function $\gamma: (0,\log d) \rightarrow (0,\log d)$ for which
\begin{equation}\label{eq:optimalbounddef}
H_\alpha(p_X^\rho)\geq \gamma\bigl(H_\beta(p_Y^\rho)\bigr)
\end{equation}
with the property that equality can be obtained for all possible values of $H_\beta(p_Y^\rho)$.

If for some state the MU bound is saturated we call this state an {\it equality state}. The corresponding point in the entropy plane is an {\it equality point}. If an equality point exists we call the MU bound {\it tight}. The MU bound is said to be {\it optimal}, whenever it completely coincides with the optimal bound.

A Hadamard matrix is a unitary matrix $U$ with elements satisfying $|U_{jk}|=1/\sqrt{d}$. The Fourier matrix is the matrix $U^F$ with components satisfying
\begin{equation}\label{eq:Fourier}
U^F_{jk}=\frac{1}{\sqrt{d}} \ \E^{\frac{2\pi\I}{d}jk} \ , \quad j,k=0,...,d-1\;.
\end{equation}
The Fourier matrix is hence a special instance of a Hadamard matrix.
This example generalizes to the wider setting of {\it finite abelian groups}, rather than just the cyclic group of $d$ elements as in \eqref{eq:Fourier}. To this end we consider the index set $J$ for the first matrix index of $U$ to equipped with a commutative binary operation ``$+$'' turning it into a group. The second index is similarly labelled by the so-called dual group, denoted here by $K$. A symmetric way to express the relation between these groups is via the canonical bicharacter of the pair $(J,K)$, which is a function  $\bic:J\times K\to\Cx$. It has the property that the for every $k$ the function $j\mapsto\bic(j,k)$ is a homomorphism from $J$ to the complex numbers with modulus $1$, and that, conversely every such homomorphism is of this form for some unique $k\in K$. Moreover, the same is true vice versa for the functions $k\mapsto\bic(j,k)$ with fixed $j\in J$. The Fourier matrix in this case is $U_{jk}=d^{-1/2} \bic(j,k)$, where $d=\abs J=\abs K$. It is unitary and obviously a Hadamard matrix. When $d$ is not a prime there are several non-isomorphic abelian groups of order $d$.

\section{Previous results}\label{sec:previ}
There has been considerable work to generalize and improve the MU bound, e.g. by using more general entropy functions \cite{berta2010} or more than two observables \Cite{Ivanovic1992,Sanchez95,BallesterWehner2007,Wu2009} (see also \Cite{WehnerWinter2010} for a review on entropic uncertainty relations). Most efforts, however, considered only the sum of the entropies (e.g.\ \Cite{Adamczak2014,Puchala2013,Puchala2015,Rudnicki2014,krishna2001,Rastegin2010,ColesPiani2014,Zozor2014}), thereby already fixing the functional form of the tradeoff relation and not capturing all the information contained in the entropy diagram.

In this work we are instead interested in characterising the curve of minimal entropies which we consider the optimal lower bound on the two entropies involved. There are, to the best of our knowledge, only very few results in the literature about the curve of minimal entropies in the finite-dimensional setting. In \Cite{sanches98,ghirardi2003} the authors note that the MU bound is not optimal in the simplest case of dimension $d=2$ and compute the optimal bound for general unitary operators, but restricted to the Shannon case $\alpha=\beta=1$. In \Cite{englert2008} a conjecture about the entropy minimizing states is presented. We will see that this conjecture needs improvement.

\section{Equality in the Maassen-Uffink uncertainty relation}
\label{sec:equality}
The MU bound provides a lower bound on the sum of two \Renyi\ entropies that satisfy the duality relation \eqref{eq:duality}. When characterising the curve of minimal entropies, it is natural to first investigate the case of equality in the MU bound.
If the unitary operator linking the observables is a Hadamard matrix, it is clear that the MU bound is tight. Indeed, any eigenvector of the observables, $\{x_i\}$ or $\{y_i\}$, is an equality state. But can one also find equality points for arbitrary unitary operators?

There already exist some results in the literature discussing this question, most importantly \Cite{maassen2007} and \Cite{coles2011}. In the latter work the authors show the link between the two concepts of uncertainty principle and data processing inequality. Using this link the characterisation of all states that saturate the uncertainty relation reduces to the question of characterising all states for which the application of a certain channel does not imply loss of information. Employing this technique the authors can characterize all quantum states that saturate the MU bound in the restricted setting of observables related by Fourier transformation and Shannon entropies.
A more general result was obtained in \Cite{maassen2007}, namely a complete characterisation of all equality points in the special case $\alpha=\beta=1$, i.e. for Shannon entropies. Here we present an alternative proof of the uncertainty relation which allows us to generalize these from Shannon entropies to the case of arbitrary pairs of \Renyi\ entropies that satisfy the duality relation.

The main result of this section is the following Theorem. In its formulation the ``support'' of a probability distribution is the set of points with non-zero probability, and $\abs M$ denotes the number of elements of a set $M$.

\begin{thm}\label{thm:equality}
Let $\alpha,\beta>\frac{1}{2}$ be such that ${1/\alpha}+{1/\beta}=2$, and let $X,Y$ be bases with $c=\max_{j,k}\abs{\braket{y_k}{x_j}}$.
Let $\rho$ be a state, and denote by $s_X$ and $s_Y$ the supports of the distributions $p_X^{\rho}$ and $p_Y^{\rho}$.
Then equality in the MU uncertainty relation
\begin{equation}\label{EntrUnc}
H_{\alpha}(p_X^{\rho})+H_{\beta}(p_Y^{\rho})\ge  \log \frac1{c^2}
\end{equation}
is reached if and only if $\rho=\kettbra\psi$ is a pure state and, possibly after multiplying the basis vectors $x_i,y_j$ with suitable phases, the following condition holds:
\begin{equation}\label{psieq}
  {\braket{x_i}\psi}= |s_X|^{-1/2}\ ,\quad {\braket{y_j}\psi}= |s_Y|^{-1/2}, \AND
  \braket{y_j}{x_i}=c \quad\mbox{for}\ i\in s_X\ \mbox{and}\ j\in s_Y.
\end{equation}
Moreover,
\begin{equation}\label{eq:easyCheckCriterion}
  \abs{s_X}\, \abs{s_Y}= \frac1{c^2}.
\end{equation}
\end{thm}

\begin{proof}
We assume first that  $\rho=\kettbra\psi$ is pure, and will show that this choice is even necessary at the end of the proof. We fix $\psi$ from now on, and choose phases for the basis elements so that, for $i\in s_X$, $j\in s_Y$ we have
\begin{equation}\label{psipos}
  \ph_i=\braket{x_i}\psi>0\AND \phhi_j=\braket{y_j}\psi>0.
\end{equation}
Note that this will change neither $c$ nor the probability distributions. Furthermore, we assume without loss of generality that $\alpha\leq\beta$. We usually eliminate $\beta$ by the duality relation, so the basic parameter to choose is $\alpha$ with $1/2<\alpha\leq1$.

Our proof is inspired by interpolation theory, and involves the application of the maximum principle to a certain analytic ``magic function'' $F$. We do not pretend that finding this function is straightforward, since we also came by it in several stages of generalization and simplification. We define
\begin{eqnarray}\label{magic}
  F(z)&=& c^{1-z}\lambda^{z}\sum_{i,j\in s} \ph_i^{\alpha z}\
              \braket{x_i}{y_j}\ \phhi_j^{\beta z}\\ \WITH
  \lambda&=& \bigl(\norm{\ph^\alpha}_2\,\norm{\phhi^\beta}_2\bigl)^{-1},
\end{eqnarray}
\piccaption{
\label{fig:strip} Domain of $F$ in the complex plane
}
\parpic(0.5\textwidth,7cm)(1.2cm,7cm)[r]{
\def\svgwidth{0.45\textwidth}
\begingroup%
  \makeatletter%
  \providecommand\color[2][]{%
    \errmessage{(Inkscape) Color is used for the text in Inkscape, but the package 'color.sty' is not loaded}%
    \renewcommand\color[2][]{}%
  }%
  \providecommand\transparent[1]{%
    \errmessage{(Inkscape) Transparency is used (non-zero) for the text in Inkscape, but the package 'transparent.sty' is not loaded}%
    \renewcommand\transparent[1]{}%
  }%
  \providecommand\rotatebox[2]{#2}%
  \ifx\svgwidth\undefined%
    \setlength{\unitlength}{216.78515625bp}%
    \ifx\svgscale\undefined%
      \relax%
    \else%
      \setlength{\unitlength}{\unitlength * \real{\svgscale}}%
    \fi%
  \else%
    \setlength{\unitlength}{\svgwidth}%
  \fi%
  \global\let\svgwidth\undefined%
  \global\let\svgscale\undefined%
  \makeatother%
  \begin{picture}(1,0.92085569)%
    \put(0,0){\includegraphics[width=\unitlength]{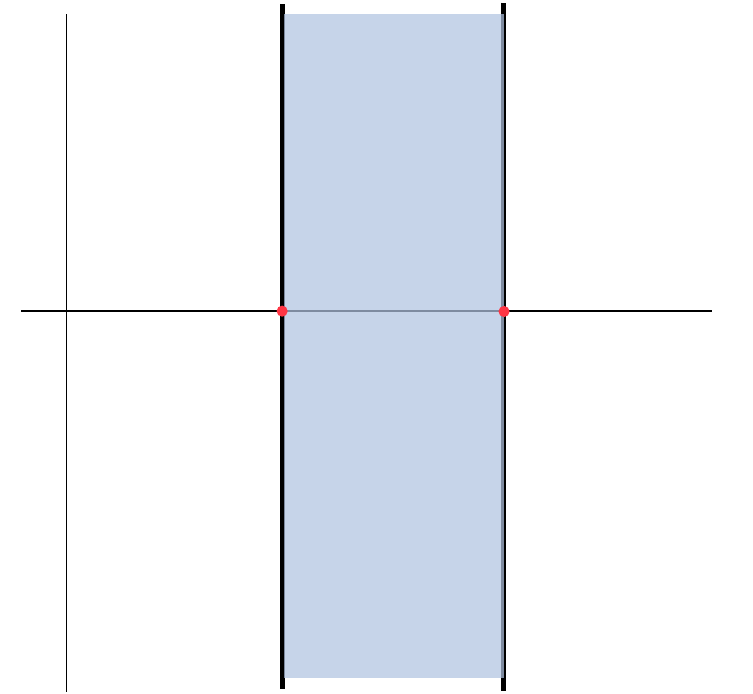}}%
    \put(0.49201546,0.71491051){\color[rgb]{0,0,0}\makebox(0,0)[lb]{\smash{$\mathcal{G}$}}}%
    \put(0.05333619,0.42907551){\color[rgb]{0,0,0}\makebox(0,0)[lb]{\smash{$0$}}}%
    \put(0.3301079,0.42907551){\color[rgb]{0,0,0}\makebox(0,0)[lb]{\smash{$1$}}}%
    \put(0.6068796,0.42907551){\color[rgb]{0,0,0}\makebox(0,0)[lb]{\smash{$2$}}}%
    \put(-0.05737249,0.83500733){\color[rgb]{0,0,0}\makebox(0,0)[lb]{\smash{$\Imag(z )$}}}%
    \put(0.82829701,0.44752692){\color[rgb]{0,0,0}\makebox(0,0)[lb]{\smash{$\Real (z )$}}}%
  \end{picture}%
\endgroup%

}
where ``$i,j\in s$'' is short hand for $i\in s_X$ and $j\in s_Y$, and $\ph^\alpha$ is the componentwise power of $\ph$, so that
\begin{equation}\label{normphmu}
  \norm{\ph^\alpha}_2^2=\sum_i \ph_i^{2\alpha},
\end{equation}
and similarly for $\phhi$. The domain $\GG$ on which this function is analyzed is the strip
\begin{equation}\label{strip}
  \GG=\{z\in\Cx\,|\,1\leq\Real z\leq2\},
\end{equation}
which is also depicted in \Fig~\ref{fig:strip}. Now since the sum \eqref{magic} is finite and $\abs{r^{\alpha z}}=r^{\alpha \Real z}$ is bounded on $\GG$ for every $r>0$, $F$ is also bounded on $\GG$, and the restriction of an entire analytic function. We claim that it is bounded in absolute value by $1$. We estimate this separately for the two boundary lines. That is, for $r\in\Rl$ we have, with $U_{ij}=\braket{x_i}{y_j}$

\begin{eqnarray}
  \abs{F(1+\iir)}&=&\lambda\left|\sum_{i,j\in s}
                     \ph_i^{\alpha (1+\iir)}\ U_{ij}\phhi_j^{\beta (1+\iir)}\right|\nonumber\\
               &=&\lambda\left|\brAAket{\ph^{\alpha (1-\iir)}}{U}{\phhi^{\beta (1+\iir)}} \right|\nonumber\\
               &\leq& \lambda \norm{\ph^{\alpha (1-\iir)}}_2 \norm{\phhi^{\beta (1+\iir)}}_2 \nonumber\\
               &=& \lambda \norm{\ph^{\alpha}}_2\, \norm{\phhi^{\beta}}_2
               =1.  \label{F1plus}
\end{eqnarray}
On the other hand,
\begin{eqnarray}
  \abs{F(2+\iir)}&=&c^{-1}\lambda^{2}\left|\sum_{i,j\in s}
                     \ph_i^{\alpha (2+\iir)}\ U_{ij}\phhi_j^{\beta (2+\iir)}\right|\nonumber\\
               &\leq& \lambda^{2}\sum_{i,j\in s}
                     \ph_i^{2}\ \left|U_{ij}/c\right|\ \phhi_j^{2\beta} \label{F2plus1}\\
               &\leq& \lambda^{2}\sum_{i,j\in s}
                     \ph_i^{2\alpha}\  \phhi_j^{2\beta} \label{F2plus2}\\
               &=& \lambda^{2} \norm{\ph^{\alpha}}_2^2\, \norm{\phhi^{\beta}}_2^2
               =1. \label{F2plus}
\end{eqnarray}
Hence, by the maximum principle, $\abs{F(z)}\leq1$ for all $z\in\GG$.

In order to relate this to entropies we consider the special value $z=1/\alpha$, which always lies in the strip, but for $\alpha=1$ is a boundary point. We get
\begin{eqnarray}
  F\left(\frac1\alpha\right)
          &=&c^{1-1/\alpha}\lambda^{1/\alpha} \sum_{ij\in s}
                    \ph_i\ U_{ij}\phhi_j^{\beta /\alpha}\nonumber\\
          &=& c^{1-1/ \alpha}\lambda^{1/\alpha}\sum_j\phhi_j^{1+\beta / \alpha} \label{argh2}     \\
          &=& c^{1-1/\alpha} \bigl(\norm{\ph^\alpha}_2^{-1/\alpha}\,\norm{\phhi^\beta}_2^{-1/\alpha}\bigr)
                    \norm{\phhi^\beta}_2^2  \label{argh3} \\
          &=& c^{1-1/\alpha} \norm{\ph^\alpha}_2^{-1/\alpha}\,\norm{\phhi^\beta}_2^{1/\beta},
\end{eqnarray}
where at \eqref{argh2} we used that $\sum_i\ph_iU_{ij}=\phhi_j$, and at \eqref{argh3} the definition of $\lambda$ and duality of $\alpha$ and $\beta$. For taking the logarithm of this expression we use that
\begin{eqnarray}
  \log\bigl(\norm{\ph^\alpha}_2^{-1/\alpha}\bigr)&=&-\frac{1-\alpha}{2\alpha}H_\alpha(\ph^2)
  \nonumber\\\AND
  \log\bigl(\norm{\phhi^\beta}_2^{1/\beta}\bigr)&=&\frac{1-\beta}{2\beta}H_\beta(\phhi^2)
                           =H_\beta(\phhi^2)
\end{eqnarray}
and get, equivalently to $F(1/\alpha)\leq1$, the inequality
\begin{equation}\label{logF}
  \log F\left(\frac1\alpha\right)= -\frac{1-\alpha}{2\alpha}
                        \Bigl(\log(c^2)+H_\alpha(\ph^2)+ H_\beta(\phhi^2)\Bigr) \leq0.
\end{equation}
For $\alpha\neq1$ we cancel the common factor and get the MU inequality. For $\alpha=1$ we always get $F(1)=1$, and the MU inequality is obtained by taking the limit $\alpha\to1$. However, it is better to express it instead by the derivative of $F$. For $\alpha=\beta=1$ we get
\begin{equation}\label{Fprime1}
  F'(1)=-\log c -\frac12 H_1(\ph^2)-\frac12 H_1(\phhi^2)\leq0,
\end{equation}
because for small $\veps$ we must have $F(1+\veps)\leq1$.

The advantage of this derivation of the MU inequality is that we have powerful characterizations of the equality case. So suppose that equality holds in the MU inequality.  Then for $\alpha<1$ this means that $F$ attains its maximum modulus $1$ at the interior point $1/\alpha$ of the strip $\GG$, and the Phragm\'en-Lindel\"of Theorem\cite{phragment} tells us that $F=1$ is the constant function. For $\alpha=1$ we need a variant of the maximum principle due to Hopf\cite{hopf} (see, e.g. Thm.~2.7 in \Cite{hopf2}), saying precisely that if the maximum is attained at the boundary with vanishing derivative we once again must have a constant function. In either case we conclude that $F(z)=1$ for all $z\in\GG$.

With this information we can go back to the above estimates for \eqref{F2plus}, which must now be tight. The first step, the triangle inequality \eqref{F2plus1}, is tight if all terms in the sum have the same argument, so up to a common phase the $U_{ij}$ for $i\in s_X$ and $j\in s_Y$ must be positive. With the phase convention \eqref{psipos} this means $U_{ij}>0$ for all $i,j$ in the supports. The second estimate \eqref{F2plus2} is only tight when all $U_{ij}$ also have the maximum allowed modulus $c$. Hence $U_{ij}=c$. If we consider $U$ as an operator on vectors with support $s_Y$ it thus maps to constant functions, so $\ph$ must be constant on $s_X$. By the same token $\phhi$ must be constant on $s_Y$. Taking into account the normalizations we get all assertions of the theorem in the pure case $\rho=\kettbra\psi$.

It remains to show that all equality states must be pure.
So let $\psi$ now be any unit vector in the support of $\rho$ and $\sigma=\kettbra\psi$. Then we can write $\rho=\lambda\sigma+(1-\lambda)\rho'$ with $\lambda>0$, $\rho'$ some other state, and similar convex relationships for the probability distributions. By concavity of the entropies, $\sigma$ must also be an equality state. Moreover, by strict concavity, $\sigma$ and $\rho$  must have the same distributions  $p_X^\sigma=p_X^\rho$ and $p_Y^\sigma=p_Y^\rho$, and hence the same supports $s_X,s_Y$. Going through the proof for the pure equality state $\kettbra\psi$, and in particular adopting the phase conventions made for $\psi$ we find that $U_{ij}=c$ for all $i\in s_X$ and $j\in s_Y$. But then, if we apply $U$ to any other unit vector $\psi'$ in the support of $\rho$ we find that $U\psi'$ is constant on its support $s_Y$. Hence $\psi'$ equals $\psi$ up to a phase, the support of $\rho$ is one-dimensional, and $\rho$ must be pure.

An alternative proof of the necessity of purity, at least for the Shannon case $\alpha=\beta=1$, is via inequality\cite{berta2010}
\begin{equation}\label{berta}
 H(p_X^\rho)+H(p_Y^\rho)\geq
   \log\frac1{c^2}+ H(\rho).
\end{equation}
Clearly, for equality states the correction term, the von Neumann entropy  $H(\rho)$, has to vanish, i.e., the state must be pure.
\end{proof}

An immediate consequence of Theorem~\ref{thm:equality} is that for most overlap matrices no equality states exist, because $1/c^2$ is not an integer. Since the rows of a unitary matrix must be normalized, this integer is at most $d$, in which case we must have a Hadamard matrix. When $1/c^2<d$ one can build examples with equality by first solving a unitary matrix completion problem, starting from the known $s_x\times s_Y$ block. One then has to modify the matrix by unitary rotations on the complementary blocks so that all matrix elements become $\leq c$. The lowest-dimensional example is $2=1/c^2<d=3$, and the overlap matrix
\begin{equation}
U=\left[\begin{array}{rrr}
       a&       a&   \ 0\\
       b&      -b&   \ a\\
       -b&      b&   \ a
\end{array}\right]
\WITH a=\frac1{\sqrt2}\AND b={\frac12}.
\end{equation}
Some higher-dimensional examples can be generated by replacing the matrix elements $a$ and $b$ by $aU_1$ and $bU_2$, where $U_1,U_2$ are any Hadamard matrices of the same dimension.

By definition, Hadamard matrices have $d$ orthogonal equality states with supports $(\abs{s_X},\abs{s_Y})=(1,d)$ and $(d,1)$, respectively. In prime dimension this is clearly the only possibility. However, even if the dimension is composite there may be no more than this, as the example\cite{HadCat}
\begin{equation}
C_6=\frac{1}{\sqrt{6}}\left[\begin{array}{rrrrrr}
       1&       1&       1&       1&       1&     1\\
       1&      -1&      -\eta&    -\eta^2&     \eta^2&     \eta\\
       1& -\eta^{-1}&       1&     \eta^2&    -\eta^3&   \eta^2\\
       1& -\eta^{-2}&  \eta^{-2}&      -1&     \eta^2&  -\eta^2\\
       1&  \eta^{-2}& -\eta^{-3}&  \eta^{-2}&       1&    -\eta\\
       1&  \eta^{-1}&  \eta^{-2}& -\eta^{-2}& -\eta^{-1}&    -1\\
\end{array}\right]
\end{equation}
with $\eta=\frac{1-\sqrt{3}}{2}+\I\sqrt{\frac{\sqrt{3}}{2}}$, shows. Here one can mechanically check that none of the 300  $3\times2$-submatrices has the property that all elements become equal after multiplication of rows and columns with suitable phases. Hence from Theorem~\ref{thm:equality} it is clear that the point $(\log 3,\log2)$ on the MU-line is not accessible for any state.

\begin{figure}[hb]\centering
	\includegraphics[width=\PICWIDTH]{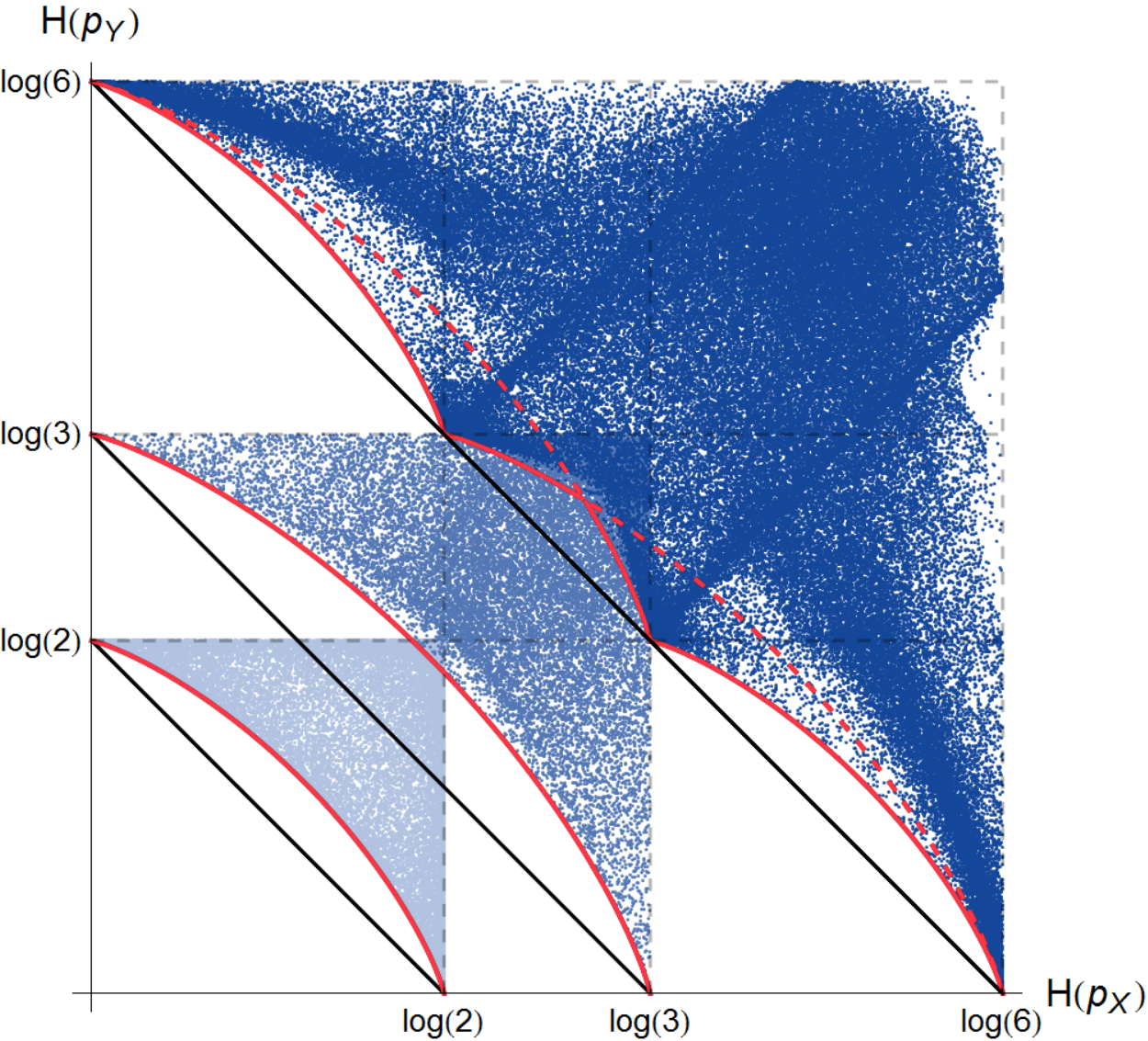}
	\caption{ \label{fig:out236} Numerical sampling of the entropy diagram for dimensions $d=2$ (light shading), $d=3$ (medium shading) and $d=6$ (dark shading) for Fourier-related observables and Shannon entropies. By Theorem \ref{thm:equality} the number of equality states corresponds to the number of divisors of the respective dimension. The optimal bounds (solid curves) are obtained by applying Conjecture \ref{conj:productNaff} and Conjecture \ref{conj:Fourier} presented in Sect.~\ref{sec:conj}.}
\end{figure}

In the special case of a Fourier matrix (see the end of Sect.~\ref{sec:prelims} for notations) we can get a complete description of the equality cases from Theorem~\ref{thm:equality}, as has been observed in Theorem~4.(1) of \Cite{coles2011} for the special case of a cyclic group. We will do the same for an arbitrary finite abelian group $J$. It turns out that the equality states are then directly linked to the subgroups  of $J$ and its dual $K$. The subgroups always come in pairs, i.e., when $L\subset J$ is a subgroup, so is its annihilator\cite{rudin}
\begin{equation}\label{annihi}
  L^\perp=\{k\in K\,|\, \forall j\in L\ \bic(j,k)=1\} \subset K.
\end{equation}
The basic result about annihilators is that $(L^\perp)^\perp=L$ for every subgroup, so there is a ono-to-one correspondence between the subgroups of $J$ and $K$, under which
$L_1\subset L_2\Leftrightarrow L_1^\perp\supset L_2^\perp$. For any non-empty set $L\subset J$,  we denote by $\chi_L$ the $\ell^2$-normalized indicator function, i.e., $\chi_L(j)=\abs L^{-1/2}$ for $j\in L$ and $\chi_L(j)=0$ otherwise.

\begin{cor}\label{lem:fouriereq}
Let $J$ be a finite abelian group, with Fourier matrix $U$, and $L\subset J$ a subgroup. Then
\begin{equation}\label{FourierH}
  U \chi_L=\chi_{L^\perp},
\end{equation}
and the vectors of the form $\chi'(j')= \bic(j',k)\, \chi_L(j'-j)$, where $j\in J/L$ and $k\in K/L^\perp$ are an orthonormal basis so that each $\kettbra{\chi'}$ is an equality state. Moreover, all equality states are of this form.
\end{cor}

Note that in the formula for $\chi'$ we can take arbitrary $j\in J$ and $k\in K$, but two such choices $(j_1,k_1)$ and $(j_2,k_2)$ define the same function $\chi'$ when $j_1-j_2\in L$ and $k_1-k_2\in L^\perp$. This observation is expressed by taking $j,k$ in the respective quotients.

We remark that, by the fundamental structure theorem of finite abelian groups, every such group is a cartesian product of cyclic groups, and has subgroups of every order which divides $d$ (see Thm.~4.3 in \Cite{gallian}). Hence the equality points on the MU line are {\it all} points $(\log d_1,\log d_2)$ with $d_1d_2=d$.

\begin{proof}
Let $\kettbra\psi$ be an equality state. The Theorem then says that for $j\in s_X$, and $k\in s_Y$
we must have $\bic(j,k)=\mu(k)\nu(j)$ for suitable phase-valued functions $\mu:s_Y\to\Cx$ and $\nu:s_X\to\Cx$. Now we can apply translations as in the construction of $\chi'$ in the Corollary to get an equality state with $0\in s_X$ and $0\in s_Y$, from which we get $\mu(k)\nu(0)=1$ and $\mu(0)\nu(j)=1$, so that the functions $\mu,\nu$ are actually constant. After applying an overall phase factor we can assume without loss of generality, that
$\bic(j,k)=1$ for $j\in s_X$, and $k\in s_Y$, and that $\psi=\chi_{s_X}$. In terms of annihilators this is expressed equivalently by $s_Y\subset s_X^\perp$ or $s_X\subset s_Y^\perp$.

When $k\in s_X^\perp$ we still have  $\bic(j,k)=1$ for $j\in s_X$. But then $(U\psi)(k)=(U\psi)(0)>0$ and we must also have $k\in s_Y$. It follows that $s_X^\perp\subset s_Y$. Combined with the already established reverse inclusion we get that $s_Y=s_X^\perp$ and, symmetrically $s_X=s_Y^\perp$. Note that since any set of the form $A^\perp$ is automatically a subgroup, we have shown that we can take $s_X=L$, $s_Y=L^\perp$ for some subgroup $L\subset J$.

We have so far only shown that $U\chi_L$ is constant on $L^\perp$, namely equal to $\sqrt{\abs L/\abs J}$, coming from the summation of $\abs L$ terms equal to $\abs L^{-1/2}$, and observing the overall normalization factor $\abs L^{-1/2}$ of the Fourier matrix. We also have to show that $\sum_{j\in L}\bic(j,k)=0$ whenever $k\notin L^\perp$. However, in that case $k$ induces a non-constant complex homomorphism on $L$, so it suffices to show that such functions add up to $0$ on any finite abelian group. However, this is immediately obvious for cyclic groups, and hence follows for arbitrary groups by the structure theorem. So we conclude that $U\chi_L$ is proportional to $\chi_{L^\perp}$, and since $U$ is unitary, it must be equal, and $\abs{L^\perp}\,\abs L=\abs J$.

Finally, let us count the translates $\chi'$ for a given subgroup. Clearly, they are orthogonal to $\chi_L$ whenever either $j+L\cap L=\emptyset$ or  $k+L^\perp\cap L\perp=\emptyset$. In other words, by taking one representative $g$ from each class in $G/H$ and also one $k$ from each class in $K/L^\perp$ we get an orthogonal family. This has $(\abs J/\abs L)\,(\abs{K}/\abs{L^\perp})=\abs J$, i.e., is an orthonormal basis.
\end{proof}

For a product of abelian groups the Fourier matrix is the tensor product of the Fourier matrices of the factors. Moreover one gets many equality states by tensoring, i.e., by taking subgroups of the form $L_1\times L_2\subset J_1\times J_2$. This additive structure is quite apparent from \Fig~\ref{fig:out236}). It is therefore useful to note that this is also true without assuming the group structure. This is shown by the following result.

\begin{cor}\label{lem:tensoreq}
Let $U_1, U_2$ be unitary operators of dimension $d_1$ and $d_2$, respectively. Suppose that for each unitary operator there exist an equality state $\sigQual^{1}$ and $\sigQual^{2}$ as characterized by Theorem \ref{thm:equality}. Then the state $\sigQual=\sigQual^{1}\otimes\sigQual^{2}$ is an equality state for the unitary operator $U_1\otimes U_2$.
\end{cor}

\begin{proof}
First, note that $\max_{j,k}\abs{(U_1\otimes U_2)_{jk}}= \max_{j,k}\abs{U_{1,jk}} \, \max_{j,k}\abs{U_{2,jk}}$. The MU relation then implies that, for any state $\sigma$ on a $d_1 \, d_2$-dimensional Hilbert space,
\begin{equation}
\noindent
H_{\alpha}(p_X^{\sigma}) + H_{\beta}(p_Y^{\sigma})\geq - 2 \log \max_{j,k}\abs{(U_1\otimes U_2)_{jk}} =  - 2 \log \max_{j,k}\abs{U_{1,jk}} \, \max_{j,k}\abs{U_{2,jk}}  \ .
\end{equation}
In particular, for the state $\sigQual=\sigQual^{1}\otimes\sigQual^{2}$, we have
\begin{align}
H_{\alpha}(p_X^{\sigQual}) + H_{\beta}(p_Y^{\sigQual}) &= H_{\alpha}(p_X^{\sigQual^{1}}) + H_{\alpha}(p_X^{\sigQual^{2}}) + H_{\beta}(p_Y^{\sigQual^{1}}) + H_{\beta}(p_Y^{\sigQual^{2}}) \nonumber\\
&=  - 2 \log \max_{j,k}\abs{U_{1,jk}} \, \max_{j,k}\abs{U_{2,jk}} \ .
\end{align}
Hence, $\sigQual$ is an equality state for $U_1\otimes U_2$.
\end{proof}

This Corollary should not be taken to suggest that {\it only} products will be equality states. For example, take the Fourier matrix of any abelian group of the form $J\times J$, which is the tensor product of two copies of the Fourier matrix of $J$. Then each subgroup $L$ with $\abs J$ elements generates a basis of equality states for the point $(\log\abs J,\log\abs J)$. These are tensor product states for the subgroup $L=\{(j,0)|j\in J\}=J\times\{0\}$. But for $H=\{(j,j)|j\in J\}$ we get a maximally entangled equality state. Again, the basic idea of this example generalizes to more general settings. If $U_1$ is any Hadamard matrix and $\overline{U_1}$ its complex conjugate, the maximally entangled vector $\psi=d^{-1/2}\sum_j\ket{jj}$ is invariant under $U=U_1\otimes\overline{U_1}$. Hence both $\psi$ and $U\psi=\psi$ belong to the equidistribution on $d$ points, and $\kettbra\psi$ is an equality state with entropies $(\log d,\log d)$, just like $\kettbra\phi$ with $\phi= d^{-1/2}\sum_j\ket{1j}$.

Perhaps one of the more surprising aspects of Theorem~\ref{thm:equality} is that neither the characterization of the equality states nor indeed the value of the lower bound depends on $\alpha,\beta$. Hence we have

\begin{cor}\label{eqindalpha}
Let $\sigQual$ be an equality state, i.e. it saturates the uncertainty relation for some $\alpha,\beta>\frac{1}{2}$ satisfying the duality relation. Then $\sigQual$ is also an equality state for all other pairs $(\alpha,\beta)$ that satisfy the duality relation, including $(\alpha,\beta)=(1/2,\infty),\ (\infty,1/2)$.
\end{cor}

\begin{figure}[h]\centering
	\includegraphics[width=\PICWIDTH]{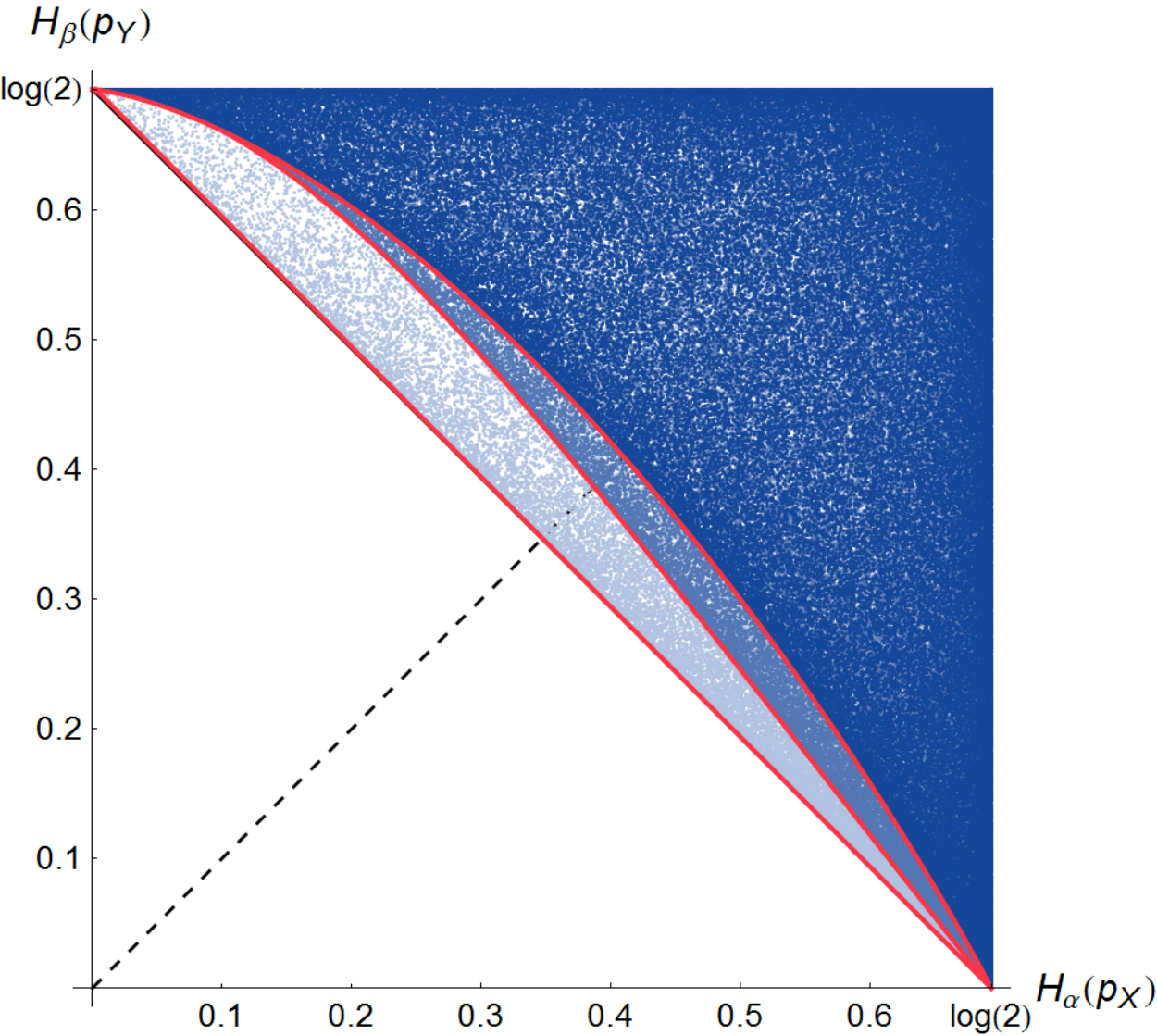}
	\caption{ \label{fig:optimalityAlpha12} Typical entropy diagram for Hadamard related observables in prime dimension for different values of $\alpha,\beta$ satisfying the duality relation \eqref{eq:duality}: $\alpha=1/2$ (light shading), $\alpha=0.6$ (medium shading) and $\alpha=0.75$ (dark shading). The MU bound is optimal if and only if $\alpha=1/2$.}
\end{figure}

The boundary cases for the inequality are proved by taking the limits on $(\alpha,\beta)$, and since the lower bound is independent of these, equality carries over. However, additional states may then also satisfy equality. Indeed,  Theorem \ref{thm:equality} does not hold in this case. As a counterexample consider an arbitrary Hadamard matrix $U$. Without loss of generality we can take it dephased, i.e., with all entries in the first row and column equal to $1/\sqrt{d}$.
Consider then some arbitrary state $\psi\in\Rl_+^d$ with real and positive components to find
\begin{equation}\label{halfinf}
  \max_k|(U\psi)_k|^2\geq|(\tilde{U}\psi)_1|^2= \frac{1}{d}\left(\sum_k\psi_k\right)^2 \ .
\end{equation}
Taking the logarithm and using the definitions \eqref{eq:defentropy} this is equivalent to
\begin{equation}\label{halfinf1}
  \log d\geq H_{\frac{1}{2}}(p_X^\psi) + H_{\infty}(p_Y^\psi),
\end{equation}
which is $\geq \log d$ by the MU inequality. Hence all such states are equality states, and we can continuously interpolate between $H_{\frac12}=0$ and $H_{\frac12}=\log d$. Thus the MU bound coincides with the optimal bound (see \Fig~\ref{fig:optimalityAlpha12}) and there is a continuum of equality states in contrast to Theorem~\ref{thm:equality}.

Another feature is true only in the boundary case, namely that for {\it every} $U$ there is an equality state.
To see this, let us consider an eigenstate $x_j$ of $X$, for which  $H_{1/2}(p_X^{x_j})=0$. But at the same time we have
\begin{equation}
\min_{j} H_{\infty}(p_Y^{x_j}) = \min_{j} (- \log \max_k |\langle y_k | x_j \rangle|^2) =  - 2 \log c \ .
\end{equation}
One could summarize this by saying that in the boundary case $\{\alpha,\beta\}=\{1/2,\infty\}$ the MU bound is just too good to allow a useful characterization of equality.

\section{Characterisation of the curve of minimal entropy pairs}\label{sec:main}
Due to the study of equality in the previous section it is clear that the MU bound is, in almost all cases, not optimal, i.e. it does not coincide with the curve of minimal entropy pairs. To characterize this optimal bound is the aim of this section. We establish three general results that hold for arbitrary dimension: First, we prove that the curve of minimal entropies can be parametrized by pure states. Second, we show that for all real-valued unitary operators we can restrict the problem to real states. And last, we establish a necessary criterion for the Fourier case which all optimal states must satisfy thereby being able to characterize a whole class of potentially optimal states. Additionally, we provide a complete characterisation of the optimal bound for the simplest case of two-dimensional state space, $d=2$. For $d=3$ there is an analytic expression\cite{englert2008}, which is well-confirmed by numerics, although not proved. However, for higher dimensions the optimal bound remains unknown. Nevertheless, we present random samples that suggest a number of conjectures, which, if true, vastly simplify the characterisation of the optimal bound.

\subsection{Sufficiency of pure states}
In this section we show that the optimal bound can be parametrized by pure states. At a first glance, this result may seem not too surprising since the situation is clear when minimizing only one concave functional $f(\rho)$ over all states:
In this case one can immediately restrict to pure states, since one of the convex components $\rho'$ of $\rho$ must always give a value $f(\rho')\leq f(\rho)$. However, the situation is not so simple when we consider a pair of concave functions, and the image of the state space under a two-component mapping $f=(f_1,f_2)$ as in (\ref{deff}).  The direct consequence of concavity is then that for, say $\rho=(\rho_1+\rho_2)/2$, the point $f(\rho)$ lies above the midpoint $M=\bigl(f(\rho_1)+f(\rho_2)\bigr)/2$ in the coordinatewise ordering, i.e., $f_i(\rho)\geq\bigl(f_i(\rho_1)+f_i(\rho_2)\bigr)/2$ for $i=1,2$ (see \Fig~\ref{fig:concave}). We therefore cannot conclude that the set $\{f(\rho)\}$ is convex: the midpoint $M$ is not in general in the set. Indeed this is clearly shown by the entropy diagrams, from which it is also clear that the complement is not convex either, except in simple cases.

\begin{figure}\centering
  \def\svgwidth{0.65\linewidth}
  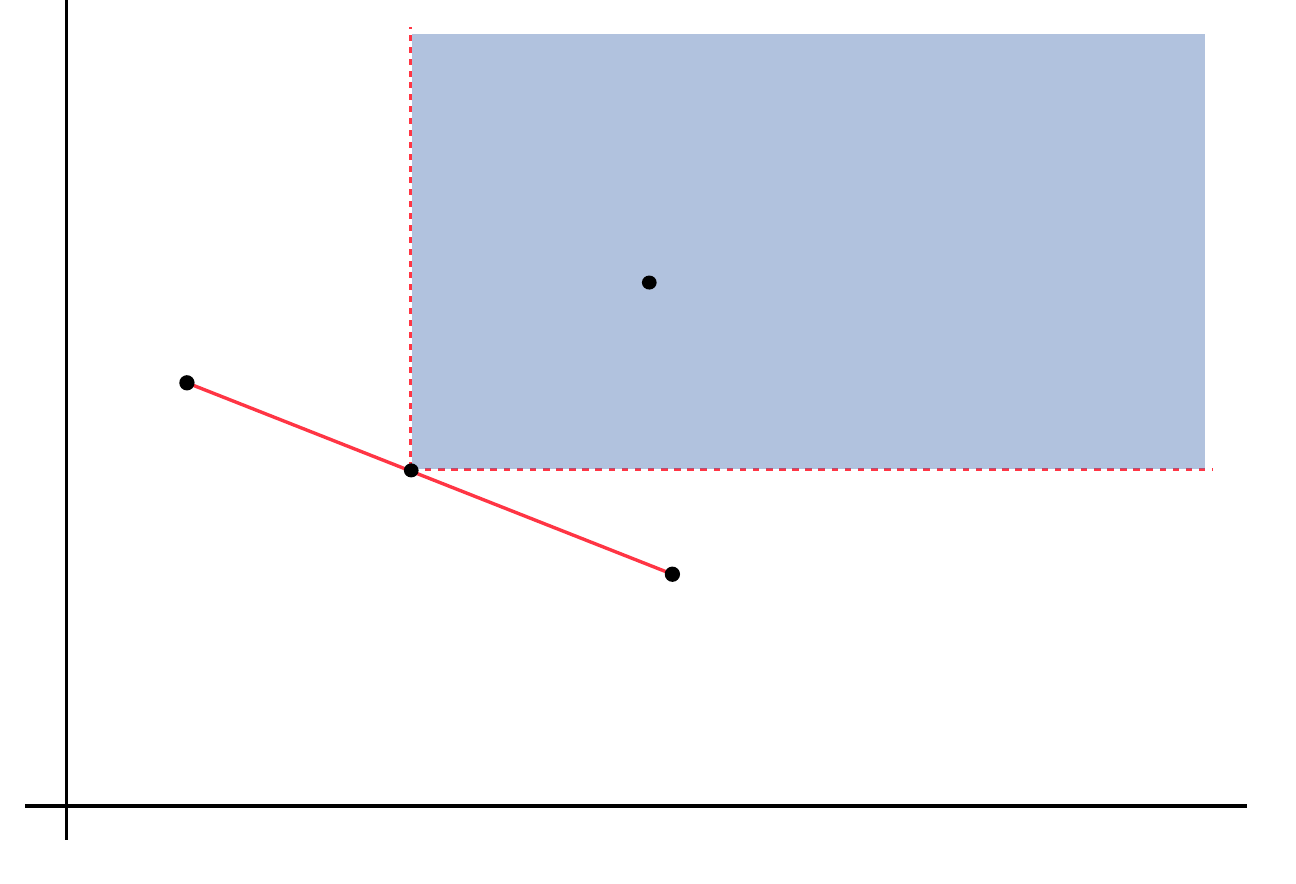
  \caption{Consequences of concavity for the set of entropy pairs.  \label{fig:concave}}
\end{figure}

For the same reasons it is not obvious that it is sufficient to restrict to pure states. This is highlighted by looking at the problem a bit more generally, considering the pairs of probability distributions in two bases.

\begin{prop}Consider two orthonormal bases $X,Y$ in a Hilbert space and let $p_X^\rho,p_Y^\rho$ denote the respective probability distributions in the state $\rho$. Then
\begin{itemize}
\item If $d=2$, then for every state $\rho$ there is pure state $\sigma$ such that $p_X^\rho=p_X^\sigma$ and $p_Y^\rho=p_Y^\sigma$.
\item If $d\leq3$, then for every $\rho$ we can find a convex decomposition $\rho=\sum_i\lambda_i\sigma_i$ into pure states $\sigma_i$ with
      $p_X^\rho=p_X^{\sigma_i}$ for all $i$.
\end{itemize}
For larger dimensions both statements fail.
\end{prop}

Thus, for $d=2$ the range $\{f(\rho)\}$ is already exhausted by pure states, and for $d=3$ the monotone closed uncertainty diagram can be computed just with pure states. For if $f(\rho)$ is any point in the diagram, we can decompose into the $\sigma_i$, without any increase of $f_1$, so by concavity we know one of the pure components has smaller $f_2$. However, this proof strategy will fail for $d\geq4$.

\begin{proof}
(1) For $d=2$, the set of quantum states $\rho$ with the same distribution $p_X^\rho$ is the intersection of the Bloch ball with a hyperplane. Intersecting with the hyperplane for $p_Y^\rho$ we get a line, which also intersects the Bloch sphere, i.e., there is a pure state with the same distributions.

(not 1) The example uses Fourier transform in $d=3$. Two density operators have the same position distribution iff their diagonals coincide and the same momentum distribution iff the sums $\sum_x\brAAket x\rho{x+y}$ coincide for all $y$. Now consider a diagonal matrix with diagonal entries $(1,1,0)/2$. A pure state with this diagonal will have just one non-zero phase in the 1-2 matrix element, so the sum with $y=1$ will be non-zero other than for the mixed state.

(2) Let us consider the convex subset $K(p)$ of states with $p_X^\rho=p$. We have to show that for $d=3$ all extreme points of this set are, in fact, pure. Our method will also show that this fails for $d\geq4$.

First observe by just conjugating with a positive diagonal operator from right and left we get an isomorphism of $K(p)$ and $K(q)$, as long as $p,q$ have the same support (of size $d$). So we may as well take $p$ to be uniform, for which we write $K(1)$ (Normalization factors are irrelevant here).

Let us sort the potential extreme points by rank. Full rank is not possible, since then {\it any} vector with uniform distribution could be subtracted with a positive weight. Rank 1 is uninteresting, because it is of the form we want to exclude. This takes care of $d=2$ and leaves only the rank 2 case for $d=3$.

So let us consider the case of rank 2 for general $d$. Let $\phi_1,\phi_2$ be two linearly independent vectors in the range of the density operator $\rho=\kettbra{\phi_1}+\kettbra{\phi_2}$. The condition that $\rho$ has uniform position distribution means that $|\phi_1(x)|^2+|\phi_1(x)|^2=1$ for all $x$. In other words, the pair $\Phi(x)=\bigl(\phi_1(x),\phi_2(x)\bigr)\in\Cx^2$ is a unit vector for every $x$ .
Then we ask whether there is any non-zero vector $\Psi\in\Cx^d$ of the form $\Psi(x)=\overline{\alpha_1}\phi_1(x)+ \overline{\alpha_2}\phi_2(x)$ such that $|\Psi(x)|=1$ for all $x$. This would be a convex component of $\rho$ with even distribution, so we could further decompose $\rho$.

We can read this as a scalar product $|\langle\alpha,\Phi(x)\rangle|^2$. Think of the $\Phi(x)$ and of $\alpha$ as represented on the Bloch sphere, where the geodesic distance is just a function of the above scalar product. So our question reduces to: Given $d$ vectors on the sphere, can we find one further vector which has the same distance from each of them?

Now for $d=2$ this is obvious, and for $d=3$ it works just like in the planar geometry of triangles: The locus of all points which have the same distance from $\Phi(1)$ and $\Phi(2)$  is a great circle bisecting their connecting geodesic at a right angle. Intersect with the bisector for $\Phi(2)$ and $\Phi(3)$, which gives a point which has the same distance from all three points. Therefore, for $d=3$, there are no extreme points of rank 2, hence all are of rank 1 as claimed.

For higher $d$ it is easy to find $d$ points, which do not lie on a circle, i.e., there is no point equidistant from all of them. Hence there are extreme points of $K(1)$ of rank 2.

\end{proof}

Surprisingly however, pure states can be shown to saturate all uncertainty diagrams, practically without assumptions on $X,Y,\alpha,\beta$.

\begin{thm}\label{pureNaff}
Let $f_1,f_2$ be continuous concave functionals on the state space, define the order relation $\eless$ as after equation (\ref{deff}). Then for every state $\rho$ there is a pure state $\sigma$ such that
$\sigma\eless\rho$.
\end{thm}

\begin{proof}
The plan of the proof is to show that for every non-pure $\rho$ we can find another state $\sigma$ of strictly smaller rank such that $\sigma\eless\rho$. Then we can successively lower the rank, arriving finally at a pure state.

\begin{figure}
\centering
  \def\svgwidth{0.65\linewidth}
  \large
  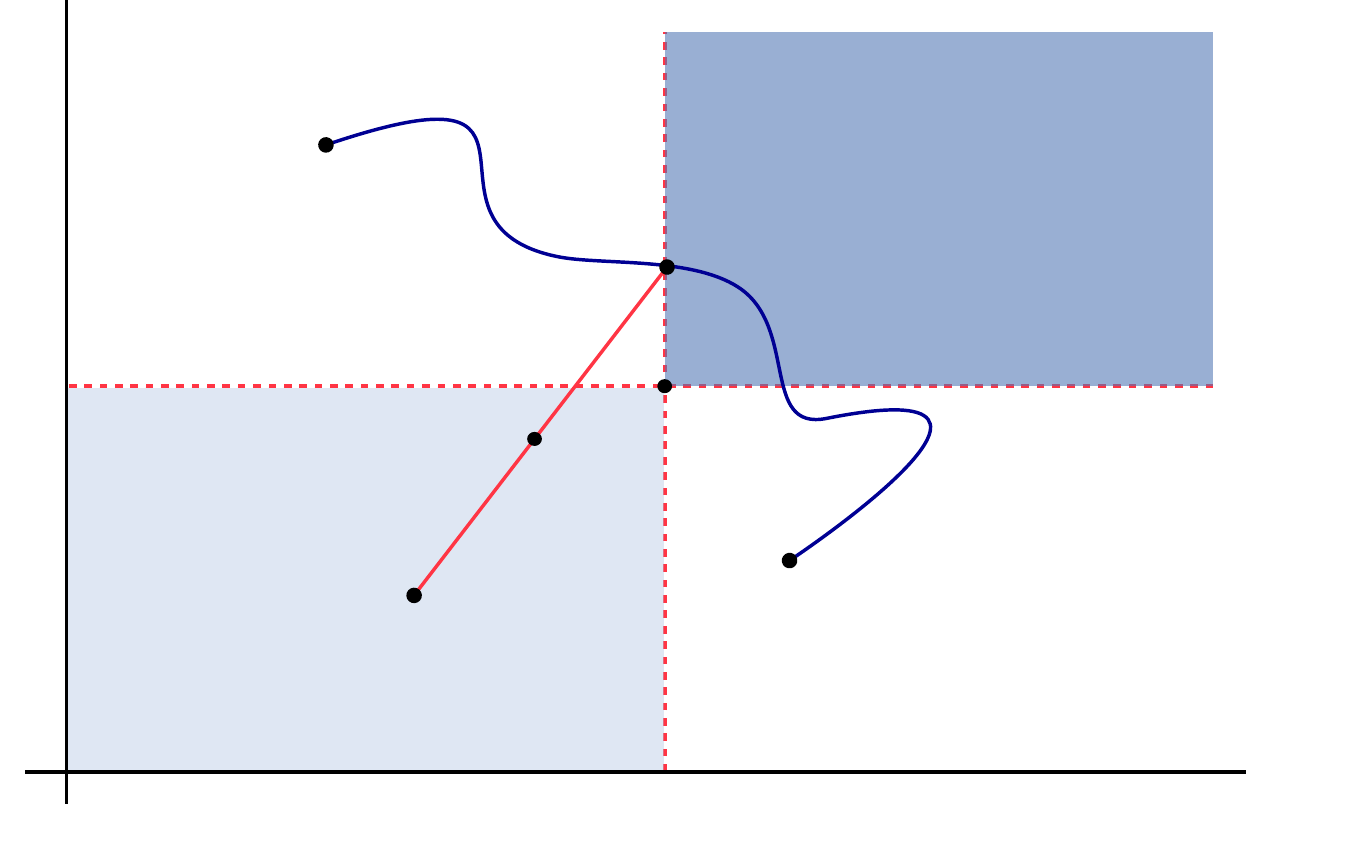
  \caption{States appearing in the proof of Theorem~\ref{pureNaff} as mapped to the two entropies plane. \label{fig:mfs}}
\end{figure}

Consider the face $F$ of the state space generated by $\rho$. Its topological boundary $\partial F$ consists precisely of the possible convex components of $\rho$ of lower rank, and is connected. For each point $\sigma\in\partial F$ there is a unique ``antipode'' $\sigma^\antipode$. It is defined as
\begin{equation}\label{antipode}
    \sigma^\antipode=\frac1\lambda\Bigl(\rho-(1-\lambda)\sigma\Bigr)
\end{equation}
for the smallest $\lambda$ for which the right hand side is positive semidefinite. It is clearly a state of reduced rank, i.e., $\sigma^\antipode\in\partial F$. We note that the required weight $\lambda$ cannot be $0$ or $1$.

We need not consider the case that $\sigma\eless\rho$, since otherwise we have found the desired element. Therefore, by exchanging the functions $f_1$ and $f_2$ if necessary, we may assume that $f_1(\sigma)>f_1(\rho)$.
We cannot also have $f_1(\sigma^\antipode)\geq f_1(\rho)$. Indeed, this would lead to the contradiction
\begin{equation}
    f_1(\rho)\geq (1-\lambda)f_1(\sigma)+\lambda f_1(\sigma^\antipode) >f_1(\rho).
\end{equation}

Now consider a continuous curve $[0,1]\ni t\mapsto\gamma(t)\in\partial F$ connecting $\sigma$ and $\sigma^\antipode$, i.e., such that $\gamma(0)=\sigma$ and $\gamma(1)=\sigma^\antipode$ (see \Fig~\ref{fig:mfs}). Since $f_1$  was assumed to be continuous the previous argument shows that, for some $t$, $f_1\bigl(\gamma(t)\bigr)=f_1(\rho)$.

If $f_2\bigl(\gamma(t)\bigr)\leq f_2(\rho)$ we have found the desired element $\gamma(t)\eless\rho$.
The non-trivial case to consider is therefore $f_2\bigl(\gamma(t)\bigr)>f_2(\rho)$, or $\rho\eless\gamma(t)$.
Let $\lambda\in(0,1)$ be the weight so that $\rho=(1-\lambda)\gamma(t)+\lambda\gamma(t)^\antipode$. Then by concavity, for $i=1,2$,
\begin{eqnarray}
    f_i(\rho)&\geq& (1-\lambda)f_i\bigl(\gamma(t)\bigr)+\lambda f_i(\gamma(t)^\antipode) \nonumber\\
             &\geq& (1-\lambda)f_i(\rho)+\lambda f_i\bigl(\gamma(t)^\antipode\bigr) \nonumber\\
    \mbox{ i.e.,}\quad
    f_i(\rho)&\geq&  f_i\bigl(\gamma(t)^\antipode\bigr).     
\end{eqnarray}
Therefore $\gamma(t)^\antipode\eless\rho$.

\end{proof}

\subsection{Sufficiency of real states for real unitary matrices}

From the previous section we know that for all unitary operators the complete optimal bound can be parametrized by pure states. Now we show that if the unitary matrix linking the two observables is real-valued, then we can further restrict the set of states for the complete optimal bound to the set of real-valued vectors. In this whole subsection we fix the Hilbert space to be $\Cx^d$ with componentwise complex conjugation, so that the real vectors $\Rl^d\subset\Cx^d$ are naturally embedded. 

\begin{thm}\label{thm:real}
Let $f_1,f_2$ be continuous concave functionals on the state space and their inputs linked by a real unitary operator $U_{\rm real}$. Also define the order relation $\eless$ as after equation (\ref{deff}). Then for every state $\rho$ there is a pure and real state $\sigma$ such that
$\sigma\eless\rho$.
\end{thm}

\begin{proof}
The idea of the proof is to employ again the proof technique of Theorem \ref{pureNaff}, i.e. decompose a state in two states with the desired property (in this case, real states) and use the concavity property of the functions.

Let $\psi\in\Cx^d$ be a pure state. Since we are interested in a decomposition into real states, it is natural to consider the decomposition 
\begin{equation}
\psi=\sqrt{\lambda} v + \I \sqrt{1-\lambda} w
\end{equation}
where $v,w\in{\Rl}^d$ are the normalized real and imaginary part of $\psi$ and $\lambda=|\Real(\psi)|^2$ ranges from $0$ to $1$. We are only interested in the case where neither $v\eless\psi$ nor $w\eless\psi$, otherwise the statement follows immediately. Furthermore, we assume without loss of generality that $f_1(v)>f_1(\psi)$. Similar to the proof in Theorem \ref{pureNaff} we cannot also have that $f_1(w)>f_1(\psi)$ because we would then find the contradiction
\begin{equation}
f_1(\psi)\geq\lambda f_1(v) + (1-\lambda)f_1(w)>f_1(w) \ .
\end{equation}
Consider now the states
\begin{equation}
\ph(t):=\E^{\I t}\psi
\end{equation}
and their normalized real and imaginary part
\begin{align}
\gamma(t)&:=\Real\bigl(\ph(t)\bigr)/|\Real\bigl(\ph(t)\bigr)|\ ,\nonumber\\
\sigma(t)&:=\Imag\bigl(\ph(t)\bigr)/|\Imag\bigl(\ph(t)\bigr)|
\end{align}
such that
\begin{equation}
\ph(t)= \sqrt{\mu(t)} \gamma(t) + \I \sqrt{1-\mu(t)}\sigma(t) \ ,
\end{equation}
where $\mu(t)=||\gamma(t)||$.
Note that $f_i(\ph(t))=f_i(\psi)$ for all $t\in(0,2\pi)$. Also note that for a real-valued unitary operator the probability distributions $p_{X}^{\ph(t)}$ and $p_{Y}^{\ph(t)}$ have the same form
\begin{equation}
p_{X/Y}^{\ph(t)}= \mu(t) p_{X/Y}^{\gamma(t)} + \bigl(1-\mu(t)\bigr)p_{X/Y}^{\sigma(t)} \ .
\end{equation}
Due to continuity we know that there exists $t_0$ such that either $\gamma(t_0)\eless\psi$, from which we obtain the desired statement, or $\psi\eless\gamma(t_0)$. Using the concavity of the functions $f_i$, the latter then implies
\begin{align}
f_i(\psi)=f_i\bigl(\ph(t_0)\bigr)&\geq \mu(t_0) f_i\bigl(\gamma(t_0)\bigr) + \bigl(1-\mu(t_0)\bigr)  f_i\bigr(\sigma(t_0)\bigl)\nonumber\\
&\geq  \mu(t_0) f_i(\psi) + \bigl(1-\mu(t_0)\bigr)  f_i\bigr(\sigma(t_0)\bigl) \ ,
\end{align}
from which obtain $f_i\bigl(\sigma(t_0)\bigr)\leq f_i(\psi)$, or equivalently $\sigma(t_0)\eless\psi$.
\end{proof}

\subsection{Variatonal method}
So far we characterized the optimal bound by the order relation $\eless$. Equivalently, we may also consider an optimisation problem as mentioned in \eqref{eq:optimalbounddef}: Given some fixed value of $H_{\beta}(p_Y^{\rho})= \delta$ the optimal bound $\gamma$ is described by minimising $H_{\alpha}(p_X^{\rho})$, i.e.
\begin{equation}
\gamma(\delta)=\min_{\rho} \{ H_{\alpha}(p_X^{\rho}) |H_{\beta}(p_Y^{\rho})= \delta \} \ ,
\end{equation}
where $\delta$ ranges from $0$ to $\log d$. However, performing this optimisation is in general quite difficult, especially  because a nice characterisation of the constant entropy set $\{\rho | H_{\beta}(p_Y^{\rho})= \delta \}$ is not known. Instead, we restrict to optimising over a subset of this constant entropy set, namely states with varied phases. Clearly, this method will not yield a sufficient criterion for a state to be optimal. However, it provides us with a necessary criterion which allows us to identify a whole class of candidates of optimal states.

More concretely, using Theorem \ref{pureNaff} we consider pure states $\ph\in\Cx^d$ and denote the components of the phase-varied state in $Y$ basis by
\begin{equation}
\psi_j=\ph_j \exp \left( \frac{2\pi\I}{d} \, \theta_j \right)
\end{equation}
for some phases $\theta_j$. Varying these phases does not change the probability distribution, $p_Y^{\psi}=p_Y^{\ph}$, and hence the phase varied states form a subset of the constant entropy set. For observables linked by Fourier transformation, we can optimize $H_{\alpha}(p_X^{\psi})$ over these states to find the following extremality criterion:

\begin{lem}
Let the two observables $X$ and $Y$ be linked by the Fourier matrix \eqref{eq:Fourier} and let $\psi$ denote an optimal state of this setup. Furthermore, let $\hat{\psi}$ denote the Fourier transform of $\psi$. Then $\psi$ satisfies
\begin{equation}\label{eq:extremality}
\Imag\left( \psi_k \sum_{j=1}^d \frac{\partial H_{\alpha}(p_X^{\psi})}{\partial |\hat{\psi}_j|^2} \overline{\hat{\psi}_j} \exp\Big(\frac{2\pi\I j k}{d} \Big)  \right)=0 \quad \forall k .
\end{equation}
\end{lem}

\begin{proof}
In order to optimize $H_{\alpha}(p_X^{\psi})$ we compute
\begin{equation}
\left.\frac{\partial H_{\alpha}(p_X^{\psi})}{\partial \theta_k}\right|_{\theta=0} = \left.\sum_{j=1}^d \frac{\partial H_{\alpha}(p_X^{\ph})}{\partial |\hat{\psi}_j|^2} \frac{\partial |\hat{\psi}_j|^2}{\partial \theta_k}\right|_{\theta=0} \overset{!}{=} 0 \ . \label{eq:proofextremality1}
\end{equation}
With $\omega:=\exp\big(\frac{2\pi\I}{d}\big)$ the Fourier transform of $\psi$ is defined as $\hat{\psi}_j:=\frac{1}{\sqrt{d}}\sum_{m=1}^d \psi_m \omega^{jm}$ and, hence,
\begin{equation}
|\hat{\psi}_j|^2=\frac{1}{d}\sum_{m,n=1}^d\ph_m\overline{\ph_n}\ \omega^{j(m-n)+\theta_m-\theta_n} \ .
\end{equation}
Therefore we have
\begin{align}
\left.\frac{\partial |\hat{\psi}_j|^2}{\partial \theta_k}\right|_{\theta=0} &= \left.\frac{1}{d}\sum_{m,n=1}^d \ph_m\overline{\ph_n}\omega^{j(m-n)+\theta_m-\theta_n} \right|_{\theta=0}\nonumber \\
&= \frac{2\pi \I}{d^2} \Imag\big(\ph_k \overline{\hat{\ph}_j} \omega^{jk}  \big) \label{eq:proofextremality2}
\end{align}
Combining \eqref{eq:proofextremality1} and \eqref{eq:proofextremality2} we obtain the desired statement.
\end{proof}
Any optimal state must necessarily satisfy the above criterion. This allows us to characterize a whole class of potentially optimal states:

\begin{lem}
	Let $\ph$ be a real-real symmetric state, i.e. a real state, $\ph\in\Rl^d$, satisfying the symmetry condition
	\begin{equation}\label{eq:rrsState}
	\ph(j)=\ph(d-j) \quad \forall j=1,...,d-1
	\end{equation}
	or, equivalently, a real state with real Fourier transform, $\hat{\ph}\in\Rl^d$. Then $\ph$ satisfies the extremality criterion \eqref{eq:extremality}.
\end{lem}

\begin{proof}
We first note a simple, but important property of real-real symmetric states: If $\ph$ is a real-real symmetric state and $\xi$ is a state with components $\xi_j=f(\ph_j)$, where $f$ is any function taking real numbers to real numbers, then $\xi$ is also a real-real symmetric state. For example, the Fourier transform of any real-real symmetric state is also real-real symmetric.

Now $\ph$ is assumed to be real-real symmetric. Hence, $\hat{\ph}$ is real-real symmetric. Define
\begin{equation}
\xi_j:=\frac{\partial H_{\alpha}(p_X^{\psi})}{\partial |\hat{\ph}_j|^2} \hat{\ph}_j
\end{equation}
and note that $\xi$ is also real-real symmetric. Importantly this implies that its Fourier transform, $\hat{\xi}$ is real.
We therefore have for all $k$
\begin{equation}
\Imag\left( \ph_k \sum_{j=1}^d \frac{\partial H_{\alpha}(p_X^{\psi})}{\partial |\hat{\ph}_j|^2} \overline{\hat{\ph}_j} \exp\Big(\frac{2\pi\I j k}{d} \Big)  \right) = \Imag\left( \sum_{j=1}^d \xi_j \exp\Big(\frac{2\pi\I j k}{d} \Big)  \right)=\Imag\left(\hat{\xi}\right) =0 \ ,
\end{equation}
which finishes the proof.
\end{proof}
	
\subsection{Simplest case: $d=2$}\label{sec:d2}
The results we presented so far are not sufficient to provide a complete characterisation of the curve of minimal entropy pairs. In what follows we therefore restrict to small dimension in order to reduce the complexity of the problem.

\begin{figure}[ht]\centering
	\includegraphics[width=\PICWIDTH]{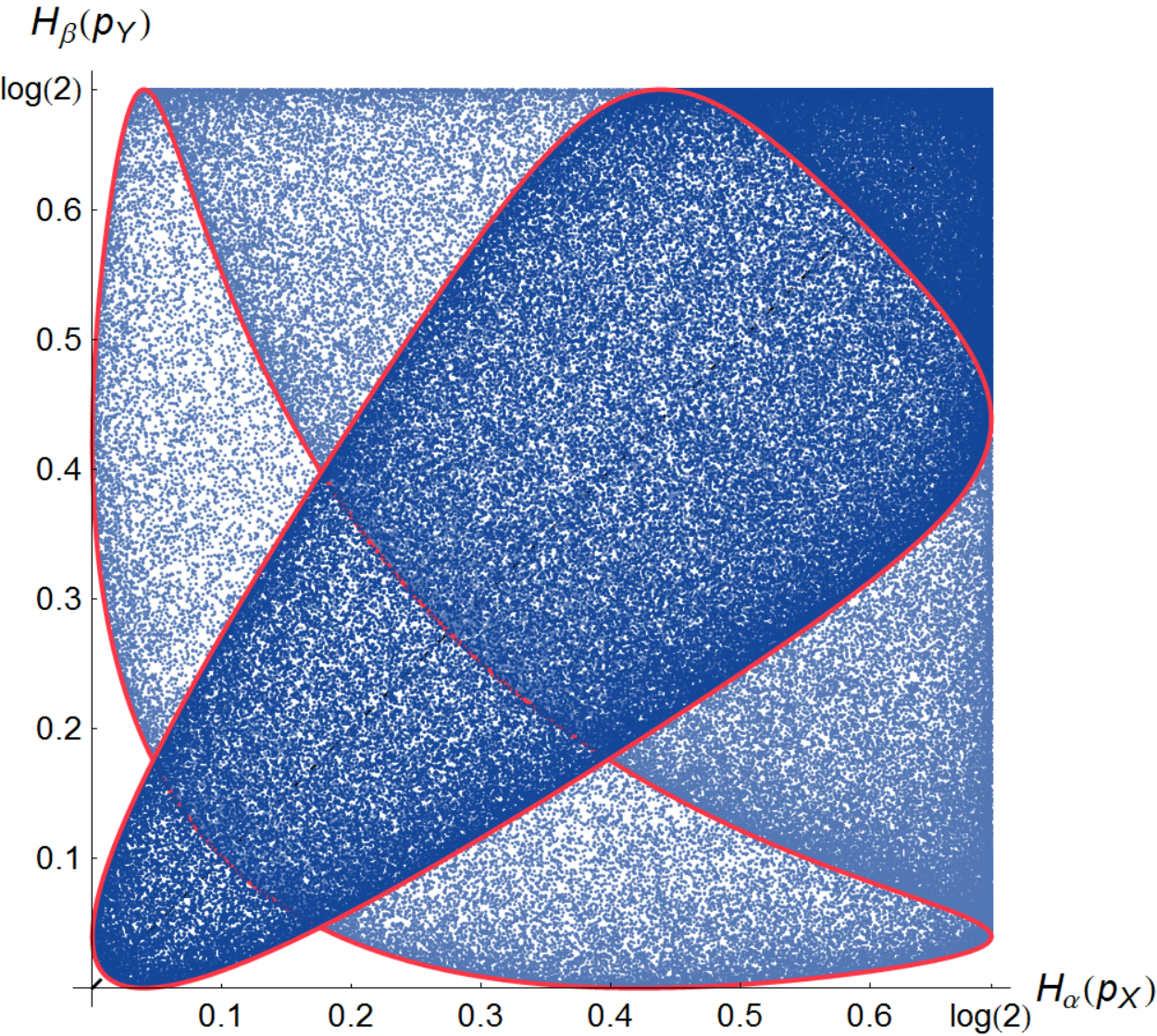}
	\caption{ \label{fig:d2arbitraryUnitaryArbitraryAlphaBeta} The optimal bound can be completely characterized in the qubit case (solid curves). The plot illustrates two entropy diagrams for randomly chosen unitary operators and entropy-pairs with $\alpha=\beta=10$ (light shading) and $\alpha=\beta=8$ (dark shading).}
\end{figure}

More concretely, we investigate the simplest case, where the dimension of the Hilbert space is $d=2$. In \Cite{sanches98,ghirardi2003} the authors characterized the curve of minimal entropy pairs for all unitary operators while restricting to the case of Shannon entropies. We now generalize their result to arbitrary pairs of \Renyi\ entropies: First we show that for each $2\times 2$ unitary operator $U$ there is a real unitary operator $\tilde{U}$ with the same entropy diagram. Then from Theorem \ref{thm:real} we can immediately infer that the lower bound can be parametrized by real states. More concretely, our aim is to show that any unitary operator, which we can always write in $\{x_i\}$ basis up to an (irrelevant) global phase as
\begin{equation}
U=\left(\begin{array}{cc}
\cos(\ph)&       \sin(\ph) \E^{-\rm  i\theta}      \\
-\sin(\ph) \E^{\I \theta}&       \cos(\ph)
\end{array}\right)  \ ,
\end{equation}
is equivalent to the matrix
\begin{equation}
\tilde{U}=\left(\begin{array}{cc}
\cos(\ph)&       \sin(\ph)     \\
-\sin(\ph)&       \cos(\ph)
\end{array}\right) \ .
\end{equation}
Indeed, the entropy diagram does not change if we first modify the unitary operator to $U'=U V$ if $V$ is a unitary operator satisfying $Vx_i=\exp(\I\ph_i)x_i$ for some phases $\ph_i$ and all $i$, since then for any state $\rho$ there exists a state $\rho'$ that yields the same pair of entropies. To see this, let $\rho'=V^{\dagger}\rho V$ to find that
\begin{equation}
p_X^{\rho'}(i)=\brAAket{x_i}{\rho'}{x_i}=\brAAket{x_i}{V^{\dagger}\rho V}{x_i}=\brAAket{x_i}{\rho}{x_i}=p_X^{\rho}(i)
\end{equation}
and
\begin{equation}
p_{Y'}^{\rho'}(j)=\brAAket{y'_j}{\rho'}{y'_j}=\brAAket{y_j}{V V^{\dagger}\rho V V^{\dagger}}{y_j}=\brAAket{y_j}{\rho}{y_j}=p_{Y}^{\rho}(j) \ .
\end{equation}
Now consider the modification
\begin{equation}
U'=\left(\begin{array}{cc}
\cos(\ph)&       \sin(\ph)       \\
-\sin(\ph) \E^{\rm  i\theta}&       \cos(\ph) \E^{\rm  i\theta}
\end{array}\right)
\end{equation}
obtained via the unitary operator
\begin{equation}
V_{\theta}=\left(\begin{array}{cc}
1&      0     \\
0&      \E^{\rm  i\theta}
\end{array}\right)
\end{equation}
However, $U'$ yields exactly the same probability distribution as $\tilde{U}$. Hence, by Theorem \ref{thm:real} the curve of minimal entropy pairs can be parametrized by real states.

Since the real states form a one-parameter family it is not difficult to check that the states
\begin{equation}\label{eq:optd2}
\psi=\bigl(\cos(\xi),\sin(\xi)\bigr) \ ,
\end{equation}where the range of $\xi$ is either $(0,\arccos(|U_{1,1}|)$ or $(\arccos(|U_{1,1}|),\pi/2)$ depending on whether $\arccos(|U_{1,1})|\in (\pi/4,3\pi/4)$ or not, parametrize the curve of minimal entropy pairs for all unitary operators and all \Renyi\ entropies. The problem is therefore completely solved in the simplest case $d=2$ (see \Fig~\ref{fig:d2arbitraryUnitaryArbitraryAlphaBeta}).

\subsection{Numerical sampling and conjectures}
\label{sec:conj}
In the previous section we characterized the optimal bound in the special case of dimension $d=2$. To the best of our knowledge the problem is unsolved for all other dimensions. Instead the authors of \Cite{englert2008} provide a conjecture stating that the curve of minimal entropies is traced out by states of the form
\begin{equation}\label{eq:Englert}
\psi= (\sqrt{p_2},\sqrt{p_2},...,,\sqrt{p_2},\sqrt{p_1})^{\mathsf T}\
\end{equation}
with $p_1 + (d-1) p_2=1$ in the case of complex Hadamard matrices and Shannon entropies. Due to the results of \Cite{sanches98,ghirardi2003} it is clear that this conjecture is correct for $d=2$. The conjecture also holds true in the case $d=3$ if we trust the numerics presented in \Fig~\ref{fig:out236}, where the solid curve directly corresponds to the  states \eqref{eq:Englert}. However, for $d=4$ we show that the conjecture already fails: For complex Hadamard matrices $c=1/\sqrt{d}$ and, hence, according to our analysis of equality in the MU bound there must be three distinct equality points, whereas the conjectured states only yield two equality points (see \Fig~\ref{fig:out248}).

\begin{figure}[t]\centering
	\includegraphics[width=\PICWIDTH]{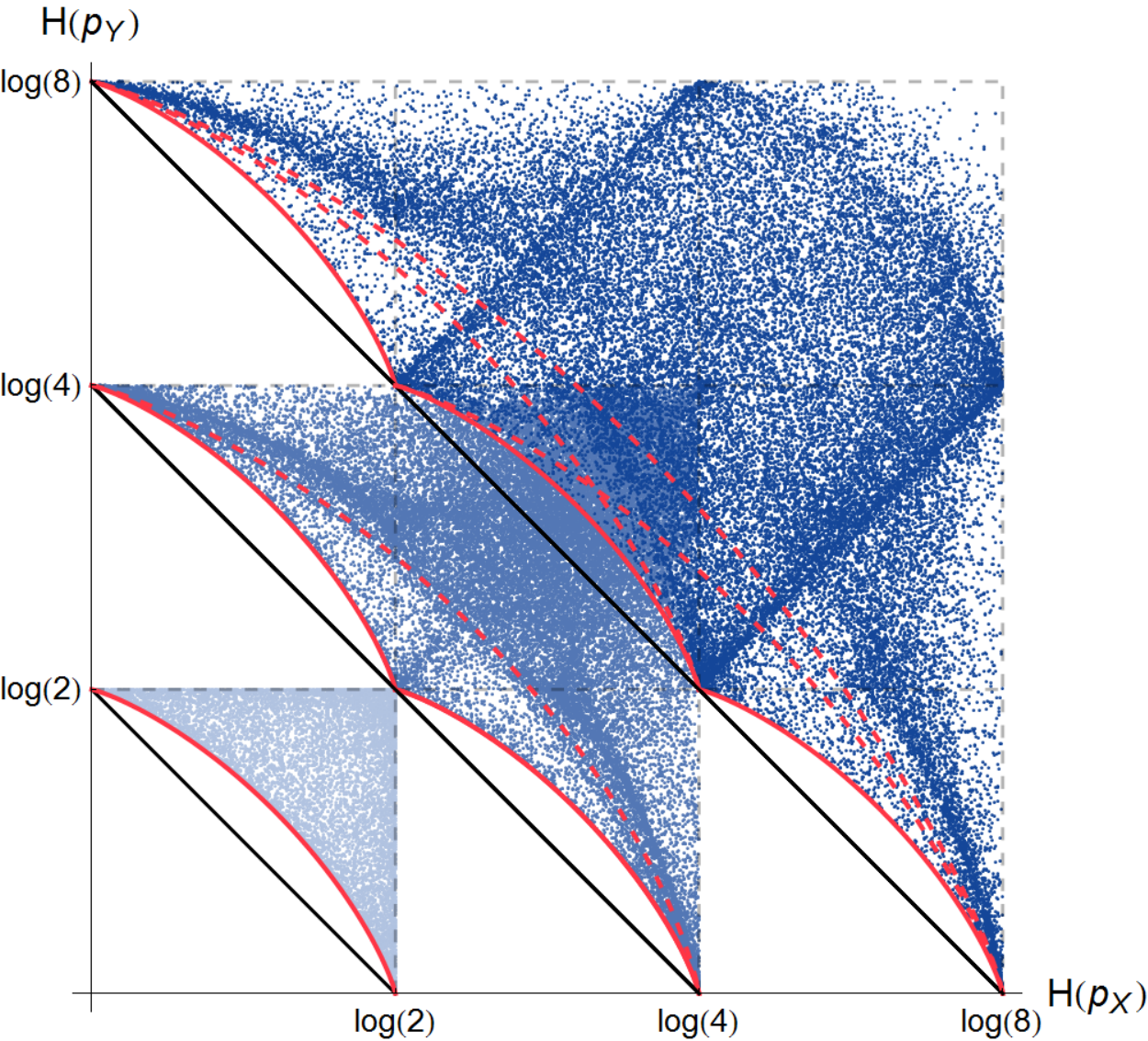}
	\caption{\label{fig:out248} Random sample of the entropy diagram for dimensions $d=2$ (light shading), $d=4$ (medium shading) and $d=8$ (dark shading) for Fourier related observables and Shannon entropies. Our results falsify a previous conjecture by Englert {\it et al.} (dashed curves). Instead the optimal bounds are given by the solid curves, which are obtained by applying Conjecture \ref{conj:productNaff} and Conjecture \ref{conj:Fourier}.}
\end{figure}
However, we present two different conjectures which, if correct, explain how the bound in \Fig~\ref{fig:out236} and \ref{fig:out248} can be obtained:
\begin{conjecture}
 \label{conj:productNaff}
	(Product states for matrices with product form)\\
	Let the unitary operator $U$ linking the two observables be a matrix of the form $U=U_1\otimes U_2$. Then for any state $\rho$ there exists a product state $\rho_1\otimes\rho_2$ with the same pair of entropies.
\end{conjecture}

The consequence of this our first conjecture is that the curve of minimal entropies for product form unitary operators in some composite dimension $d=d_1\, d_2$ is just comprised of tensor products of states that parametrize the curve in dimension $d_1$ and $d_2$, respectively. Indeed, from the additivity of the \Renyi\ entropy it then directly follows that a state $\rho_d=\rho_{d_1}\otimes \rho_{d_2}$ is optimal with respect to the unitary operator $U=U_{d_1}\otimes U_{d_2}$ if and only if the marginals $\rho_{d_1}$ and $\rho_{d_2}$ are optimal with respect to the unitary operators $U_{d_1}$ and $U_{d_2}$, respectively.
We note that this conjecture also agrees with our findings for the equality states, especially Corollary \ref{lem:tensoreq}.
\begin{conjecture}
 \label{conj:Fourier}
	(Decomposition of the Fourier matrix)\\
	Let the two observables be linked by the Fourier matrix $U^F_d$ of composite dimension $d=d_1\, d_2$. Then the entropy diagram does not change if we replace $U^F_d$ by $U^F_{d_1}\otimes U^F_{d_2}$.
\end{conjecture}
		
The consequence of this second conjecture is that, although the Fourier matrix can, in general, not be decomposed into a tensor product of Fourier matrices of smaller dimension, the entropy diagram (and hence the curve of minimal entropy pairs) does not change under this replacement. Hence, if this conjecture were true, we could apply Conjecture \ref{conj:productNaff} and characterize the curve of minimal entropy pairs by states of product form, where the marginals parametrize the optimal bound in the respective smaller dimension.

As an example let us consider Fourier related observables in dimension $d=4$. Employing both conjectures we know that it suffices to consider only the problem of characterising the optimal bound for Fourier related observables in dimension $d=2$. But for such observables we already characterized the bound completely (see Sect.~\ref{sec:d2}) and, hence, the optimal bound in $d=4$ is traced out by product states with marginals given by \eqref{eq:optd2}. Indeed, this result agrees with the random sample (\Fig~\ref{fig:out248}). In \Fig~\ref{fig:out236} we also show other examples, where the numerics validate the two conjectures above.

Note that the above conjectures are statements about the case of composite dimension, effectively stating that for a large class of unitary operators one only needs to solve the problem in prime dimension. The prime-dimensional case, however, still remains a hard problem. But we can provide two further conjectures that, if correct, vastly reduce the complexity of calculating the optimal bound in these instances:

\begin{conjecture}
 \label{conj:independence}
	(Independence of the optimal states of $(\alpha,\beta)$)\\
If $\rho$ is an optimal state for any unitary operator and any $\alpha,\beta>\frac12$ satisfying the duality relation \eqref{eq:duality}, then $\rho$ is also an optimal state for all other dual pairs.
\end{conjecture}

This conjecture can be seen as an extension of Corollary \ref{eqindalpha}. Note that we again excluded the extremal case $\{\alpha,\beta\}=\{1/2,\infty\}$ for the same reasons as explained in Sect.~\ref{sec:equality}. In \Fig~\ref{fig:optimalityAlpha12} the optimal bounds, although differently shaped, are traced out be the same states which supports Conjecture \ref{conj:independence}.

The last conjecture only considers the case of observables linked by the Fourier matrix.

\begin{conjecture}
 \label{conj:realrealNaff}
	(Sufficiency of real-real symmetric states for Fourier)\\
	If $\rho$ is an optimal state for the Fourier case, then there is a real-real symmetric state $\sigma$ as given by \eqref{eq:rrsState} with the same entropy pair.
\end{conjecture}

\begin{figure}[t]\centering
	\includegraphics[width=\PICWIDTH]{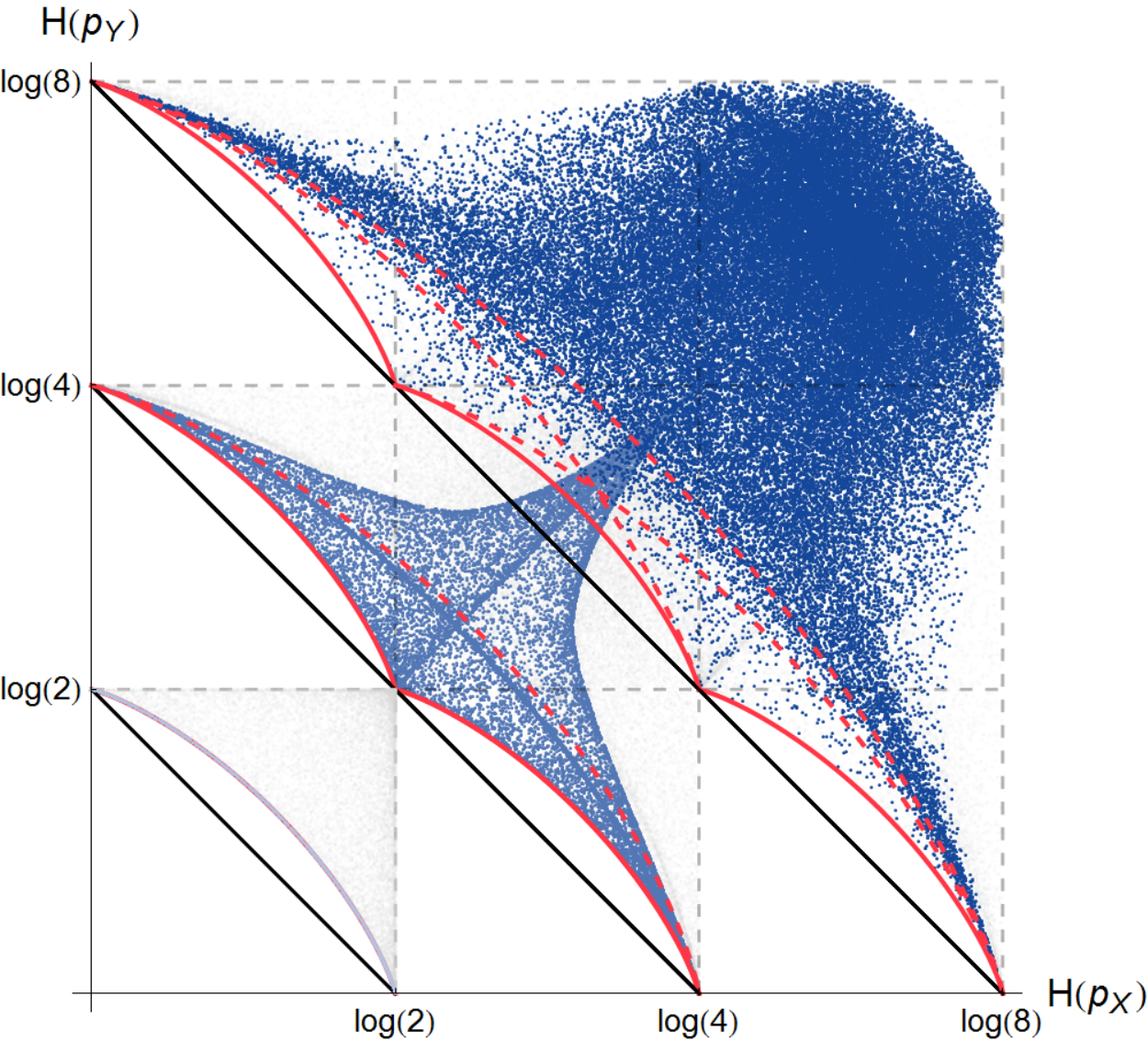}
	\caption{ \label{fig:conjectures}  Random sample of the entropy diagram for real-real symmetric states in dimensions $d=2$ (light shading), $d=4$ (medium shading) and $d=8$ (dark shading) for Fourier related observables and Shannon entropies. Restricting to real-real symmetric states does not yield the complete entropy diagram (grey), but seems to be sufficient to characterize the optimal bound.}
\end{figure}
According to this conjecture it is sufficient to analyse the problem only for real-real symmetric states, which yields a huge simplification in both analytical and numerical treatments of the problem. As an example consider Fourier related observables in dimension $d=3$. If Conjecture \ref{conj:realrealNaff} were correct, we already knew a characterisation of the optimal bound, since the real-real symmetric states in this case form a one-parameter family and therefore trace out the desired curve. Indeed, for $d=3$ the real-real symmetric states coincide with the states conjectured by \Cite{englert2008} which, as mentioned above, trace out the bound if we trust numerics. \Fig~\ref{fig:conjectures} also suggests the validity of Conjecture \ref{conj:realrealNaff}.

Furthermore, we note that real-real symmetric states are closed under the tensor product, in the sense that any tensor product of two real-real symmetric states is again a real-real symmetric state. Hence, Conjecture \ref{conj:productNaff} and Conjecture \ref{conj:realrealNaff} agree with each other.

\section{Conclusion and outlook}
We investigated the curve of minimal entropies that completely describes the entropic uncertainty tradeoff between two observables. We showed that the lower bound on the sum of two entropies as given by the Maassen-Uffink uncertainty relation is not optimal in almost all cases and hence does not correspond to the curve of minimal entropies. To show this, we presented a novel proof of the MU bound that allowed us to analyse the case of equality in the uncertainty relation.

In order to characterize the curve of minimal entropies, we provided three main results: First, we showed that the optimal bound can be traced out by pure states. Second, the optimal bound for real-valued unitary operators can be traced out by real-valued pure states. And last, we presented an extremality criterion, which any optimal state must satisfy. Numerical and analytical results for the case of small dimension suggest a number of conjectures that, if true, lead to a drastic reduction of the optimisation space. The optimal lower bound could then be computed.

\section*{Acknowledgements}
The work in Singapore is funded by the Singapore Ministry of Education (partly through the Academic Research Fund Tier 3 MOE2012-T3-1-009) and the National Research Foundation of Singapore.
F.F. acknowledges support from LUH GRK 1463 and the Japan Society for the Promotion of Science (JSPS) by KAKENHI grant No. 24-02793.
R.S. acknowledges support from the BMBF funded network Q.com-Q.

\bibliography{ourbib}

\end{document}